\documentclass[11pt]{article}
\usepackage{a4wide}
\usepackage{charter,eulervm}
\usepackage{makeidx}  % allows for indexgeneration
\usepackage{latexsym}
\usepackage{amsmath}
\usepackage{amssymb}
\usepackage{ifthen}
\usepackage{mathrsfs}
\usepackage{graphics}
\usepackage{color}
\usepackage{calc}

\usepackage[dvipsnames]{xcolor}
\usepackage{tikz}
\usepackage{tikz-cd}
\usetikzlibrary{calc,shapes,arrows}

\usepackage{stmaryrd}
\usepackage{amsthm}
\newcommand{\str}[1]{\ensuremath{\mathrm #1}}
\newcommand{\cls}[1]{\ensuremath{\mathcal #1}}

\renewcommand{\P}{\textup{P}}
\newcommand{\NP}{\textup{NP}}

\newcommand{\FPT}{\textup{FPT}}
\newcommand{\W}[1]{\text{$\textup{W}[#1]$}}
\newcommand{\fpt}{\textup{fpt}}

\newlength{\probwidth}
\newcommand{\npprob}[5][7]{
\begin{center}
  \normalfont\fbox{

\addtolength{\probwidth}{#1cm}\parbox{\probwidth}{\textsc{#2}\\\hspace*{1.5em}
    \begin{tabular}[t]{
     rp{#1cm}}\textit{Instance:}&#3. \\
     \textit{Parameter:}&#4. \\
     \textit{Problem:}&#5
    \end{tabular}}}
\end{center}}

\newtheorem{theorem}{Theorem}[section]
\newtheorem{lemma}[theorem]{Lemma}
\newtheorem{corollary}[theorem]{Corollary}
\newtheorem{proposition}[theorem]{Proposition}

\theoremstyle{definition}
\newtheorem{definition}[theorem]{Definition}
\newtheorem{remark}[theorem]{Remark}
\newtheorem{example}[theorem]{Example}

\newcommand{\Emb}{\textsc{Emb}}
\newcommand{\pEmb}{\ensuremath{p\textsc{-Emb}}}

\newcommand{\Hom}{\textsc{Hom}}
\newcommand{\pHom}{\ensuremath{p\textsc{-Hom}}}

\newcommand{\pColEmb}{\ensuremath{p\textsc{-Col-Emb}}}

\newcommand{\tw}{\textup{tw}}

\newcommand{\core}{\textup{core}}

\newcommand{\dist}{\textup{dist}}

\newcommand{\vv}{\textbf{\emph{v}}}

\newcommand{\id}{\textup{id}}

\newcommand{\bigmid}{\;\big|\;}
\newcommand{\Bigmid}{\;\Big|\;}

\newcommand{\floor}[1]{\left\lfloor#1\right\rfloor}

\usepackage{todonotes}

\begin{document}

\title{The Hardness of Embedding Grids and Walls}

\author{Yijia Chen\\\normalsize School of Computer Science\\
\normalsize Fudan University\\
\normalsize yijiachen@fudan.edu.cn\\
\and
Martin Grohe\\\normalsize Lehrstuhl f\"ur Informatik 7 \\
\normalsize RWTH Aachen University\\
\normalsize grohe@informatik.rwth-aachen.de \\
\and
Bingkai Lin\\\normalsize JST, ERATO, Kawarabayashi Large Graph Project\\
\normalsize National Institute of Informatics\\
\normalsize lin@nii.ac.jp}

\date{}
\maketitle
           % typeset the title of the contribution

\begin{abstract}
The dichotomy conjecture for the parameterized embedding problem states
that the problem of deciding whether a given graph $G$ from some class
$\cls K$ of ``pattern graphs'' can be embedded into a given graph $H$
(that is, is isomorphic to a subgraph of $H$) is fixed-parameter tractable
if $\cls K$ is a class of graphs of bounded tree width and $\W1$-complete
otherwise.

Towards this conjecture, we prove that the embedding problem is
$\W1$-complete if $\cls K$ is the class of all grids or the class of all
walls.
\end{abstract}

\section{Introduction}
The \emph{graph embedding} a.k.a \emph{subgraph isomorphism} problem is a
fundamental algorithmic problem, which, as a fairly general pattern matching
problem, has numerous applications. It has received considerable attention
since the early days of complexity theory (see,
e.g.,\cite{epp99a,garjoh79,mat78,ull76}). Clearly, the embedding problem is
NP-complete, because the clique problem and the Hamiltonian path or cycle
problem are special cases. The embedding problem and special cases like the
clique problem or the longest path problem have also played an important
role in the development of fixed-parameter algorithms and parameterized
complexity theory (see \cite{marpil13}). The problem is complete for the
class $\W1$ when parameterized by the size of the pattern graph; in fact,
the special case of the clique problem may be regarded as the paradigmatic
$\W1$-complete problem \cite{dowfel99,dowfel95a}. On the other hand,
interesting special cases such as the longest path and longest cycle
problems are fixed-parameter tractable \cite{aloyuszwi95,mon85}. This
immediately raises the question for which pattern graphs the problem is
fixed-parameter tractable.

Let us make this precise. An \emph{embedding} from a graph $G$ to a graph
$H$ is an injective mapping $f:V(G)\to V(H)$ such that for all edges $vw\in
E(G)$ we have $f(v)f(w)\in E(H)$. For each class $\cls K$ of graphs, we
consider the following parameterized problem.
\npprob{$\pEmb(\cls K)$}{Graphs $G$ (the \emph{pattern graph}) and
  $H$ (the \emph{target graph}), where $G\in \cls
K$}{$|G|$}{Decide whether there is an embedding from $G$ to $H$.}
Plehn and Voigt \cite{plevoi90} proved that $\pEmb(\cls K)$ is
fixed-parameter tractable if $\cls K$ is a class of graphs of bounded tree
width. No tractable classes $\cls K$ of unbounded treewidth are known. The
conjecture, which may have been stated in \cite{gro07} first, is that there
are no such classes.

\medskip
\noindent \textbf{Dichotomy Conjecture.}
%\begin{dc}
  $\pEmb(\cls K)$ is fixed-parameter tractable if and only if $\cls K$
  is a class of bounded treewidth and $\W1$-complete
  otherwise.\footnote{There is a minor issue here regarding the
    computability of the class $\cls K$: if we want to include
    classes $\cls K$ that are not recursively enumerable here then we
    need the nonuniform notion of fixed-parameter tractability
    \cite{flugro06}.}
%\end{dc}

\medskip
Progress towards this conjecture has been slow. Even the innocent-looking
case where $\cls K$ is the class of complete bipartite graphs had been open
for a long time; only recently the third author of this paper proved that it
is $\W1$-complete~\cite{lin15}.

Before we present our contribution, let us discuss why we expect a dichotomy
in the first place. The main reason is that similar dichotomies hold for
closely related problems. The first author, jointly with Thurley and
Weyer~\cite{chethuwey08}, proved the version of the conjecture for the
\emph{strong embedding}, or \emph{induced subgraph isomorphism} problem.
Building on earlier work by Dalmau, Kolaitis and Vardi~\cite{dalkolvar02} as
well as joint work with Schwentick and Segoufin~\cite{groschweseg01}, the
second author~\cite{gro07} proved that the parameterized homomorphism
problem $\pHom(\cls K)$ for pattern graphs from a class $\cls K$ is
fixed-parameter tractable if and only if the cores of the graphs in $\cls K$
have bounded tree width and $\W1$-complete otherwise.

Let us remark that there is no P vs.\ NP dichotomy for the classical
(unparameterized) embedding problem; this can easily be proved along the
lines of corresponding results for the homomorphism and strong embedding
problems using techniques from~\cite{bodgro08,chethuwey08,gro07}.

\subsection*{Our contribution}
We make further progress towards a the Dichotomy Conjecture by establishing
hardness for two more natural graph classes of unbounded tree width.
\begin{theorem}\label{thm:main}
$\pEmb(\cls K)$ is \W 1-hard for the classes $\cls K$ of all grids and all
walls.
\end{theorem}

See Section~\ref{sec:pre} and in particular Figure~\ref{fig:gridwall} for
the definition of grids and walls. Grids and walls are interesting in this
context, because they are often viewed as the ``generic'' graphs of
unbounded tree width: by Robertson and Seymour's~\cite{gm05} Excluded Grid
Theorem, a class $\cls K$ of graphs has unbounded tree width if and only if
all grids (and also all walls) appear as minors of the graphs in $\cls K$.

Just like the hardness result of the embedding problem for the class of all
complete bipartite graphs \cite{lin15}, our theorem looks simple and
straightforward, but it is not. In fact, we started to work on this right
after the hardness for complete bipartite graphs was proved, hoping that we
could adapt the techniques to grids. This turned out to be a red herring.
The proof we eventually found is closer to the proof of the dichotomy result
for the homomorphism problem~\cite{gro07} (also see~\cite{chemul14}). The
main part of our proof is fairly generic and has nothing to do with grids or
walls. We prove a general hardness result (Theorem~\ref{thm:richW1}) for
$\pEmb(\cls K)$ under the technical condition that the graphs in $\cls K$
have ``rigid skeletons'' and unbounded tree width even after the removal of
these skeletons. We think that this theorem may have applications beyond
grids and walls.

\subsection*{Organization of the paper}
We introduce necessary notions and notations in Section~\ref{sec:pre}. For
some technical reason, we need a colored version $\pColEmb$ of $\pEmb$. In
Section~\ref{sec:colemb} the problem $\pColEmb$ is shown to be \W 1-hard on
any class of graphs of unbounded treewidth. Then in
Section~\ref{sec:frames}, we set up the general framework. In particular, we
explain the notion of skeletons, and prove the general hardness theorem. The
classes of grids and walls are shown to satisfy the assumptions of this
theorem in Section~\ref{sec:richnessgridwall}. In the final
Section~\ref{sec:con} we conclude with some open problems.

\section{Preliminaries}\label{sec:pre}

A graph $G$ consists of a finite set of vertices $V(G)$ and a set of edges
$E(G)\subseteq \binom{V}{2}$. Every edge is denoted interchangeably by $\{u,
v\}$ or $uv$.  We assume familiarity with the basic notions and terminology
from graph theory, e.g., degree, path, cycle etc, which can be found in
e.g.,~\cite{die12}. By $\dist^G(u,v)$ we denote the distance between
vertices $u$ and $v$ in a graph $G$, i.e., the length of a shortest path
between $u$ and $v$.

Let $s,t\in \mathbb N$. A \emph{$(s\times t)$-grid $\str G_{s, t}$} has
\begin{eqnarray*}
V(\str G_{s, t})= [s]\times [t] & \text{and} &
E(\str G_{s, t})= \big\{(i,j)(i',j')\bigmid |i-i'|+ |j-j'|=1\big\}.
\end{eqnarray*}
And the \emph{wall $\str W_{s, t}$} of width $s$ and height $t$ is defined
by
\begin{align*}
V(\str W_{s, t}) & = \big\{v_{i,j} \bigmid \text{$i\in [s+1]$ and $j\in [t]$}\big\}
 \cup \big\{v_{i,t+1} \bigmid \text{$i\in [s+1]$ and odd $t$}\big\} \\
 & \quad \cup \big\{u_{i,j} \bigmid \text{$i\in [s+1]$ and $2\le j\le t$]}\big\}
  \cup \big\{u_{i,t+1} \bigmid \text{$i\in [s+1]$ and even $t$}\big\}, \\
E(\str W_{s, t}) & = \big\{v_{i,1}v_{i+1,1}\bigmid i\in [s]\big\} \\
 & \quad \cup \big\{v_{i,t+1}v_{i+1,t+1}\bigmid \text{$i\in [s]$ and odd $t$}\big\} \\
 & \quad \cup \big\{u_{i,t+1}u_{i+1,t+1}\bigmid \text{$i\in [s]$ and even $t$}\big\} \\
 & \quad \cup \big\{v_{i,j}u_{i,j}, u_{i,j}v_{i+1,j} \bigmid \text{$i\in [s]$ and $2\le j\le t$}\big\} \\
 & \quad \cup \big\{v_{i,j}v_{i,j+1} \bigmid \text{$i\in [s+1]$ and odd $j\in [t]$}\big\} \\
 & \quad \cup \big\{u_{i,j}u_{i,j+1} \bigmid \text{$i\in [s+1]$ and even $j\in [t]$}\big\}.
\end{align*}
Figure~\ref{fig:gridwall} gives two examples.
\begin{figure}%[h!]
\centering
\begin{tikzpicture}[
  scale=0.60,
  vertex/.style={circle,inner sep=0pt,minimum size=1.5mm,fill=black},
  ]

  \begin{scope}
    \foreach \x in {1,2,...,7}
        \foreach \y in {1,2,3,4}
           \node[vertex] (\x\y) at (\x,\y) {};

   \path (11) node[below] {\footnotesize(1,1)}
             (71) node[below] {\footnotesize(7,1)}
             (14) node[above] {\footnotesize(1,4)}
             (74) node[above] {\footnotesize(7,4)};

% the grid
 \draw (1,1)--(7,1);
 \draw (1,2)--(7,2);
 \draw (1,3)--(7,3);
 \draw (1,4)--(7,4);

 \draw (1,1)--(1,4);
 \draw (2,1)--(2,4);
 \draw (3,1)--(3,4);
 \draw (4,1)--(4,4);
 \draw (5,1)--(5,4);
 \draw (6,1)--(6,4);
 \draw (7,1)--(7,4);

 \draw (4,0) node{{\normalsize (a)}}; % $(7\times 4)$-grid}};
\end{scope}
%%%%%%%%%%%%%%

\begin{scope}[xshift=10cm,yshift=1cm]
  \foreach \x in {0,1,...,6}
      \foreach \y in {0,1}
      {
        \node[vertex] (v\x\y) at ($(1.6*\x,1.6*\y)$) {};
        \node[vertex] (vv\x\y) at ($(1.6*\x,1.6*\y)+(0,0.8)$) {};
        \node[vertex] (u\x\y) at ($(1.6*\x,1.6*\y)+(0.8,0.8)$) {};
        \node[vertex] (uu\x\y) at ($(1.6*\x,1.6*\y)+(0.8,1.6)$) {};
        \draw (v\x\y) to (vv\x\y);
        \draw (u\x\y) to (uu\x\y);
      }
      \draw (v00) to (v60);
      \draw (vv00) to (u60);
      \draw (v01) to (uu60);
      \draw (vv01) to (u61);
      \draw (uu01) to (uu61);

      \path (v00) node[left] {\small$v_{1,1}$};
      \path (v60) node[right] {\small$v_{7,1}$};
      \path (vv00) node[left] {\small$v_{1,2}$};
      \path (u00) node[below] {\small$u_{1,2}$};
      \path (u60) node[right] {\small$u_{7,2}$};
      \path (vv01) node[left] {\small$v_{1,4}$};
      \path (u61) node[right] {\small$u_{7,4}$};
      \path (uu01) node[left] {\small$u_{1,5}$};
     \path (uu61) node[right] {\small$u_{7,5}$};

      \path (5.2,-1) node{(b)}; % $(7\times 4)$-grid}};
\end{scope}

% the wall

%  \draw (11,1)--(17,1);

%  \draw (11,1)--(11,1.5);
%  \draw (12,1)--(12,1.5);
%  \draw (13,1)--(13,1.5);
%  \draw (14,1)--(14,1.5);
%  \draw (15,1)--(15,1.5);
%  \draw (16,1)--(16,1.5);
%  \draw (17,1)--(17,1.5);
% %%

%  \draw (11,1.5)--(17.5,1.5);

%  \draw (11.5,1.5)--(11.5,2);
%  \draw (12.5,1.5)--(12.5,2);
%  \draw (13.5,1.5)--(13.5,2);
%  \draw (14.5,1.5)--(14.5,2);
%  \draw (15.5,1.5)--(15.5,2);
%  \draw (16.5,1.5)--(16.5,2);
%  \draw (17.5,1.5)--(17.5,2);
% %%

%  \draw (11,2)--(17.5,2);

%  \draw (11,2)--(11,2.5);
%  \draw (12,2)--(12,2.5);
%  \draw (13,2)--(13,2.5);
%  \draw (14,2)--(14,2.5);
%  \draw (15,2)--(15,2.5);
%  \draw (16,2)--(16,2.5);
%  \draw (17,2)--(17,2.5);
% %%

%  \draw (11,2.5)--(17.5,2.5);

%  \draw (11.5,2.5)--(11.5,3);
%  \draw (12.5,2.5)--(12.5,3);
%  \draw (13.5,2.5)--(13.5,3);
%  \draw (14.5,2.5)--(14.5,3);
%  \draw (15.5,2.5)--(15.5,3);
%  \draw (16.5,2.5)--(16.5,3);
%  \draw (17.5,2.5)--(17.5,3);
% %%

%  \draw (11,3)--(17.5,3);

%  \draw (11,3)--(11,3.5);
%  \draw (12,3)--(12,3.5);
%  \draw (13,3)--(13,3.5);
%  \draw (14,3)--(14,3.5);
%  \draw (15,3)--(15,3.5);
%  \draw (16,3)--(16,3.5);
%  \draw (17,3)--(17,3.5);
% %%

%  \draw (11,3.5)--(17.5,3.5);

%  \draw (11.5,3.5)--(11.5,4);
%  \draw (12.5,3.5)--(12.5,4);
%  \draw (13.5,3.5)--(13.5,4);
%  \draw (14.5,3.5)--(14.5,4);
%  \draw (15.5,3.5)--(15.5,4);
%  \draw (16.5,3.5)--(16.5,4);
%  \draw (17.5,3.5)--(17.5,4);
% %%

%  \draw (11.5,4)--(17.5,4);

%  \draw (14,0) node{{\normalsize (b)}}; % $(6\times 6)$-wall.}};

\end{tikzpicture}
\label{fig:gridwall}
\caption{\ (a) A $(7\times 4)$-grid. \qquad
(b) A $(6\times 4)$-wall.}
\end{figure}
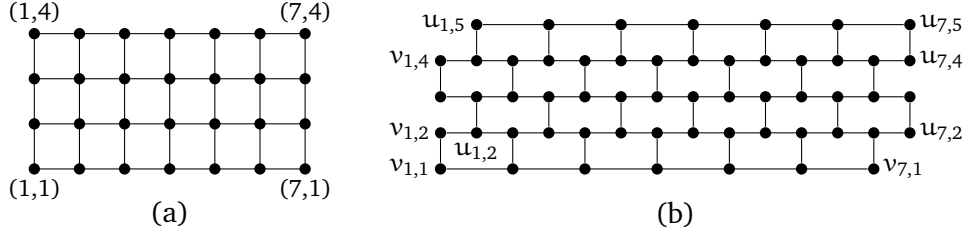

Let $G$ and $H$ be two graphs. A \emph{homomorphism} from $G$ to $H$ is a
mapping $h: V(G)\to V(H)$ such that for every edge $uv\in E(G)$ we have
$h(u)h(v)\in E(H)$. If in addition $h$ is injective, then $h$ is an
\emph{embedding} from $G$ to $H$. A homomorphism from $G$ to itself is also
called an \emph{endomorphism}, and similarly an embedding from $G$ to itself
is an \emph{automorphism}.

A \emph{subgraph} $G'$ of $G$ satisfies $V(G')\subseteq V(G)$ and
$E(G')\subseteq E(G)$. We say that $G'$ is a \emph{core} of $G$ if there is
a homomorphism from $G$ to $G'$, and if there is no homomorphism from $G$ to
any proper subgraph of $G'$. It is well known that all cores of $G$ are
isomorphic, hence we can speak of \emph{the} core of $G$, written
$\core(G)$.

Sometimes, we also consider \emph{colored} graphs in which a graph $G$ is
equipped with a coloring $\chi: V(G)\to C$ which maps every vertex to a
color in the color set $C$. We leave it to the reader to generalize the
notions of homomorphism, embedding, and core from graphs to colored graphs.
One easy but important fact is that if in the colored graph $(G, \chi)$
every vertex has a distinct color then $\core(G)= G$.

The notions of tree decomposition and treewidth are by now standard. In
particular, $\tw(G)$ denotes the treewidth of the graph $G$. For a $(s\times
t)$-grid $\str G_{s, t}$, we have $\tw(G_{s,t})= \min \{s, t\}$, and for a
$(s\times t)$-wall $\str W_{s, t}$, we have $\tw(W_{s,t})= \min \{s, t\}+1$.
The treewidth of a colored graph $(G, \chi)$ is the same as the treewidth of
the underlying uncolored graph $G$, i.e., $\tw(G, \chi)= \tw(G)$.

In a \emph{parameterized problem} $(Q, \kappa)$ every problem instance $x\in
\{0,1\}^*$ has a parameter $\kappa(x)\in \mathbb N$ which is computable in
polynomial time from $x$. $(Q, \kappa)$ is \emph{fixed-parameter tractable}
(\FPT) if we can decide for every instance $x\in \{0,1\}^*$ whether $x\in Q$
in time $f(\kappa(x))\cdot |x|^{O(1)}$, where $f: \mathbb N\to \mathbb N$ is
a computable function. Thus, \FPT\ plays the role of \P\ in parameterized
complexity. On the other hand, the so-called class \W 1 is generally
considered as a parameterized analog of \NP. The precise definition of \W 1
is not used in our proofs, so the reader is referred to the standard
textbooks, e.g.,~\cite{dowfel99,flugro06,cygfom15}. Let $(Q_1, \kappa_1)$
and $(Q_2, \kappa_2)$ be two parameterized problems. An \fpt-reduction from
$(Q_1, \kappa_1)$ and $(Q_2, \kappa_2)$ is a mapping $R: \{0,1\}^*\to
\{0,1\}^*$ such that for every $x\in \{0,1\}^*$
\begin{itemize}
\item $x\in Q_1 \iff R(x)\in Q_2$,

\item $R(x)$ can be computed in time $f(\kappa_1(x))\cdot |x|^{O(1)}$,
    where $f:\mathbb N\to \mathbb N$ is a computable function,

\item $\kappa_2(R(x))\le g(\kappa_1(x))$, where $g:\mathbb N\to \mathbb N$
    is computable.
\end{itemize}

Now we state a version of the main result of~\cite{gro07} which is most
appropriate for our purpose.
\begin{theorem}\label{thm:grohe}
Let $\cls K$ be a recursively enumerable\footnote{If $\cls K$ is not
  recursively enumerable, there is still a ``non-uniform'' hardness result. See~\cite{gro07} or a
  discussion.} class of colored graphs such that
for every $k\in \mathbb N$ there is a colored graph $(G, \chi)\in \cls K$
whose core has treewidth at least $k$. Then $\pHom(\cls K)$ is hard for \W 1
(under \fpt-reductions).
\end{theorem}

\section{From Homomorphism to Colored Embedding}\label{sec:colemb}
Let $\cls K$ be a class of graphs. We consider the following colored version
of the embedding problem for $\cls K$.
\npprob{$\pColEmb(\cls K)$}{Two graphs $G$ and $H$ with $G\in \cls K$, and a
function $\chi: V(H)\to V(G)$}{$|G|$}{Decide wether there is an embedding
$h$ from $G$ to $H$ such that $\chi(h(v))= v$ for every $v\in V(G)$.}
Thus, in the $\pColEmb(\cls K)$ problem we partition the vertices of $H$ and
associate one part with each vertex of $G$. Then we ask for an embedding
where each vertex $G$ is mapped to its part.

\begin{lemma}\label{lem:colembedW1}
Let $\cls K$ be a recursively enumerable class of graphs with unbounded
treewidth. Then $\pColEmb(\cls K)$ is hard for \W 1.
\end{lemma}

\begin{proof} Let $G$ be a graph. We expand it with the trivial coloring
$\chi_G: V(G)\to V(G)$ defined by $\chi_G(v):= v$ for every $v\in V(G)$. It
is easy to see that $\core(G, \chi_G)=(G, \chi_G)$ and hence
$\tw\big(\core(G, \chi_G)\big)= \tw(G, \chi_G)= \tw(G)$. Therefore, the
class
\[
\cls K^* := \big\{(G, \chi_G)\bigmid G\in \cls K\big\}
\]
satisfies the conditions in Theorem~\ref{thm:grohe}. Hence, $\pHom(\cls
K^*)$ is \W 1-hard, and it suffices to show $\pHom(\cls K^*)$ can be reduced
to $\pColEmb(\cls K)$ by an \fpt-reduction.

Let $(G, \chi_G)\in \cls K^*$ with $G\in \cls K$ and $(H, \chi_H)$ be a
colored graph. We construct a graph $P$ with
\begin{align*}
V(P) := & \big\{(u,v)\in V(G)\times V(H)\bigmid \chi_H(v)= u\big\}, \\
E(P) := & \Big\{(u_1,v_1)(u_2,v_2) \Bigmid (u_1,v_1),(u_2,v_2)\in V(P),\\
 & \hspace{4cm} \text{$u_1u_2\in E(G)$, and $v_1v_2\in E(H)$}\Big\}.
\end{align*}
Moreover let $\chi: V(P)\to V(G)$ be given by
\[
\chi(u,v):= u
\]
for every $(u,v)\in V(P)$. It is easy to verify that
%there is a homomorphism from $(G, \chi_G)$ to $(H, \chi_H)$ if and only if
%there is an embedding $h$ from $G$ to $P$ with $\chi(h(v))= v$ for every
%$v\in V(G)$.\qed
\begin{align*}
\text{there} & \ \text{is a homomorphism from $(G, \chi_G)$ to $(H, \chi_H)$} \\
 & \iff \text{there is an embedding $h$ from $G$ to $P$
  with $\chi(h(v))= v$ for every $v\in V(G)$}.
\end{align*}
\end{proof}

\section{Frames and Skeletons}\label{sec:frames}

Let $G$ be a graph and $D\subseteq V(G)$ such that the degree of every $v\in
D$ is at most 2, i.e., $\deg^{G}(v)\le 2$. For every $u,v\in V(G)\setminus
D$ we say they are \emph{close} (with respect to $D$) if there is a path in
$G$ between $u$ and $v$ whose internal vertices are all in $D$. We define
$G/D$ as the graph given by
\begin{align*}
V(G/D) & := V(G)\setminus D, \\
E(G/D) & := \big\{uv \bigmid \text{$u,v\in V(G)\setminus D$, $u\ne v$, and they are close}\big\}.
\end{align*}
Let $u\in D$. We say that $u$ is \emph{associated with a vertex $v\in
V(G/D)$} if $u$ is on a path in $G$ between $v$ and a vertex $w\in D$ with
$\deg^G(w)=1$ whose internal vertices are all in $D$. Similarly, $u$ is
\emph{associated} with some edge $e =vw\in E(G/D)$ if $u$ is on a path in
$G$ between $v$ and $w$ whose internal vertices are all in $D$. It should be
clear that $u$ can only be associated with a unique vertex or a unique edge
in $G/D$, and not both. Furthermore, some $w\in D$ might not be associated
with any vertex or edge; this happens precisely to all $w$ on a path or
cycle with all vertices in $D$.

\medskip
To simplify presentation, from now on we fix a graph $G$.
\begin{definition}\label{def:frame}
A set $F\subseteq V(G)$ is a \emph{frame} for $G$ if every endomorphism $h$
of $G$ with $F\subseteq h(V(G))$ is surjective.
\end{definition}

\begin{remark}\label{rem:frame}
%Let $F\subseteq F'\subseteq V(G)$ and $F$ be a frame for $G$. Then $F'$ is
%also a frame for $G$.
Let $F$ be a frame for $G$
\begin{enumerate}
\item[(1)] If $F=\emptyset$, then $G$ is a core.

\item[(2)] Any endomorphism $h$ with $F\subseteq h(V(G))$ is an
    automorphism of $G$, since $G$ is finite.

\item[(3)] Let $F'\subseteq V(G)$ with $F\subseteq F'$. Then $F'$ is also
    a frame for $G$.

\item[(4)] $V(G)$ is a frame for $G$.
\end{enumerate}
\end{remark}

\begin{definition}\label{def:skeleton}
Let $F, D\subseteq V(G)$ such that
\begin{enumerate}
\item[(S1)] $F$ is a frame for $G$,

\item[(S2)] $F\cap D= \emptyset$,

\item[(S3)] for every $v\in D$
    \[
    \deg^{G\setminus F}(v) =\big|\big\{u\in V \bigmid \text{$u\notin F$ and $\{u,v\}\in E$}\big\}\big| \le 2.
    \]
\end{enumerate}
Then we call $\mathcal S=(F,D)$ a \emph{skeleton} of $G$.
\end{definition}

\begin{example}\label{exam:gridskeleton}
Consider the grid $\str G_{7, 8}$. Lemma~\ref{lem:gridskeleton} in
Section~\ref{sec:richnessgridwall} implies that it has a skeleton $(F, D)$
with
\begin{align*}
F &= \big\{(i,j)\bigmid \text{$i\in \{1,2,6,7\}$ or $j\in \{1,7,8\}$} \big\}
  \cup \big\{(4, 2j) \bigmid j \in [3]\big\} \quad \text{and}\\
D &= \big\{(2i+1,2j) \bigmid \text{$i\in [2]$ and $j\in [3]$}\big\}
 \cup \big\{(2i,2j+1) \bigmid \text{$i\in \{2\}$ and $j\in [2]$}\big\},
\end{align*}
as shown in Figure~\ref{fig:gridskeleton}.
\begin{figure}%[h!]
\centering
\begin{tikzpicture}[scale=0.80]
 {\scriptsize

% the grid
 \begin{scope}
 \draw (1,1)--(7,1);
 \draw (1,2)--(7,2);
 \draw (1,3)--(7,3);
 \draw (1,4)--(7,4);
 \draw (1,5)--(7,5);
 \draw (1,6)--(7,6);
 \draw (1,7)--(7,7);
 \draw (1,8)--(7,8);

 \draw (1,1)--(1,8);
 \draw (2,1)--(2,8);
 \draw (3,1)--(3,8);
 \draw (4,1)--(4,8);
 \draw (5,1)--(5,8);
 \draw (6,1)--(6,8);
 \draw (7,1)--(7,8);

% the frame

 % borders
 \draw (1,1) node[circle,draw,fill=Black]{};
 \draw (1,2) node[circle,draw,fill=Black]{};
 \draw (1,3) node[circle,draw,fill=Black]{};
 \draw (1,4) node[circle,draw,fill=Black]{};
 \draw (1,5) node[circle,draw,fill=Black]{};
 \draw (1,6) node[circle,draw,fill=Black]{};
 \draw (1,7) node[circle,draw,fill=Black]{};
 \draw (1,8) node[circle,draw,fill=Black]{};

 \draw (2,1) node[circle,draw,fill=Black]{};
 \draw (2,2) node[circle,draw,fill=Black]{};
 \draw (2,3) node[circle,draw,fill=Black]{};
 \draw (2,4) node[circle,draw,fill=Black]{};
 \draw (2,5) node[circle,draw,fill=Black]{};
 \draw (2,6) node[circle,draw,fill=Black]{};
 \draw (2,7) node[circle,draw,fill=Black]{};
 \draw (2,8) node[circle,draw,fill=Black]{};

 \draw (6,1) node[circle,draw,fill=Black]{};
 \draw (6,2) node[circle,draw,fill=Black]{};
 \draw (6,3) node[circle,draw,fill=Black]{};
 \draw (6,4) node[circle,draw,fill=Black]{};
 \draw (6,5) node[circle,draw,fill=Black]{};
 \draw (6,6) node[circle,draw,fill=Black]{};
 \draw (6,7) node[circle,draw,fill=Black]{};
 \draw (6,8) node[circle,draw,fill=Black]{};

 \draw (7,1) node[circle,draw,fill=Black]{};
 \draw (7,2) node[circle,draw,fill=Black]{};
 \draw (7,3) node[circle,draw,fill=Black]{};
 \draw (7,4) node[circle,draw,fill=Black]{};
 \draw (7,5) node[circle,draw,fill=Black]{};
 \draw (7,6) node[circle,draw,fill=Black]{};
 \draw (7,7) node[circle,draw,fill=Black]{};
 \draw (7,8) node[circle,draw,fill=Black]{};

 \draw (3,1) node[circle,draw,fill=Black]{};
 \draw (4,1) node[circle,draw,fill=Black]{};
 \draw (5,1) node[circle,draw,fill=Black]{};

 \draw (3,7) node[circle,draw,fill=Black]{};
 \draw (4,7) node[circle,draw,fill=Black]{};
 \draw (5,7) node[circle,draw,fill=Black]{};

 \draw (3,8) node[circle,draw,fill=Black]{};
 \draw (4,8) node[circle,draw,fill=Black]{};
 \draw (5,8) node[circle,draw,fill=Black]{};

 % middle points in the frame

 \draw (4,2) node[circle,draw,fill=Black]{};
 \draw (4,4) node[circle,draw,fill=Black]{};
 \draw (4,6) node[circle,draw,fill=Black]{};

% the set $D$

 %the rows
 \draw (3,2) node[circle,draw,fill=lightgray]{};
 \draw (5,2) node[circle,draw,fill=lightgray]{};

 \draw (3,4) node[circle,draw,fill=lightgray]{};
 \draw (5,4) node[circle,draw,fill=lightgray]{};

 \draw (3,6) node[circle,draw,fill=lightgray]{};
 \draw (5,6) node[circle,draw,fill=lightgray]{};

 %the columns
 \draw (4,3) node[circle,draw,fill=lightgray]{};
 \draw (4,5) node[circle,draw,fill=lightgray]{};

% the remaining vertices

 \draw (3,3) node[circle,draw,fill=white]{$a$};
 \draw (5,3) node[circle,draw,fill=white]{$b$};

 \draw (3,5) node[circle,draw,fill=white]{$c$};
 \draw (5,5) node[circle,draw,fill=white]{$d$};

 \draw (4,0) node{{\normalsize (a)}};

 \end{scope}

%%%%%%%%%%%%%%

 \begin{scope}[xshift=9cm,yshift=0cm]
 \draw (1,3)--(3,3);
 \draw (1,5)--(3,5);

 \draw (1,3)--(1,5);
 \draw (3,3)--(3,5);

 \draw (1,3) node[circle,draw,fill=white]{$a$};
 \draw (3,3) node[circle,draw,fill=white]{$b$};

 \draw (1,5) node[circle,draw,fill=white]{$c$};
 \draw (3,5) node[circle,draw,fill=white]{$d$};

 \draw (2,0) node{{\normalsize (b)}};
 \end{scope}

 }

\end{tikzpicture}
\caption{(a)~A skeleton for $\str G_{7,8}$, where $F$ is the set of black vertices
and $D$ the set of light gray vertices.
(b)~The graph $(G\setminus F)/D$.}\label{fig:gridskeleton}
\end{figure}
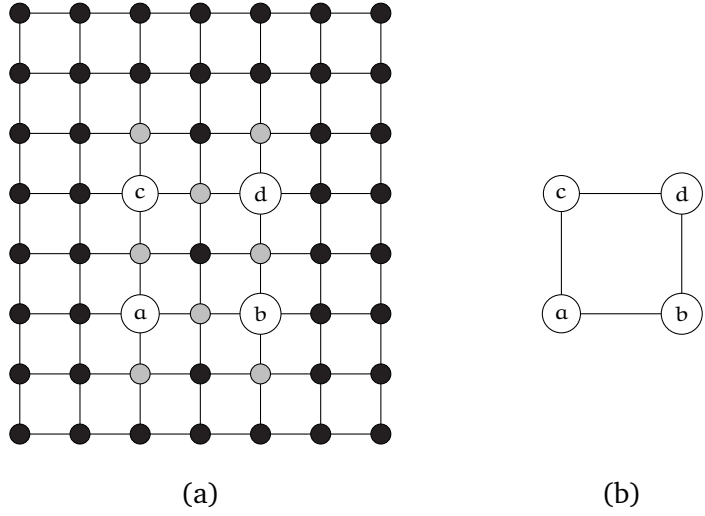
\end{example}

\begin{definition}\label{def:product}
Let $\mathcal S= (F,D)$ be a skeleton of $G$. For every graph $H$ and every
mapping
\[
\chi: V(H) \to V(G)\setminus (F\cup D),
\]
we construct a product graph $P= P(G, \mathcal S, H, \chi)$ as follows.
\begin{enumerate}
\item[(P1)] The vertex set is $V(P):= \bigcup_{i\in [4]} V_i$ with
    \begin{align*}
    V_1 & = \big\{(u,a) \bigmid \text{$u\in V(G)\setminus (F\cup D)$ and $a\in V(H)$ with $\chi(a)= u$}\big\}, \\
    V_2 & = \Big\{(u,u) \Bigmid \text{$u\in F$
     or \big($u\in D$ without being associated} \\
     & \hspace{3cm} \text{with any vertex or edge in $(G\setminus F)/D$\big)} \Big\}, \\
    V_3 & = \big\{(u, \vv_{u,a}) \bigmid \text{$u\in D$, $a\in V(H)$, and $\chi(a)= v$} \\
            & \hspace{3cm} \text{with $u$ being associated with $v$ in $(G\setminus F)/D$}\big\}, \\
    V_4 & = \big\{(u, \vv_{u,e}) \bigmid \text{$u\in D$, $e= \{a,b\}\in E(H)$, $\chi(a)= v$, and $\chi(b)=w$} \\
            & \hspace{3cm} \text{with $u$ being associated with $\{v,w\}$ in $(G\setminus F)/D$}\big\}.
    \end{align*}
    Note that in the definition of $V_3$ and $V_4$ all $\vv_{u,a}$ and
    $\vv_{u,e}$ are fresh elements.

\item[(P2)] The edge set is $E(P):= \bigcup_{1\le i\le j\le 4} E_{ij}$
    with
    \begin{align*}
    E_{11} & = \big\{(u,a)(v,b) \bigmid \text{$(u,a), (v,b)\in V_1$, $uv\in E(G)$, and $ab\in E(H)$}\big\}, \\
    E_{12} & = \big\{(u,a),(v,v) \bigmid \text{$(u,a)\in V_1$, $(v,v)\in V_2$, and $uv\in E(G)$}\big\}, \\
    E_{13} & = \big\{(u,a)(v, \vv_{v,a}) \bigmid \text{$(u,a)\in V_1$, $(v, \vv_{v,a})\in V_3$,
     and $uv\in E(G)$}\big\}, \\
    E_{14} & = \big\{(u,a)(v, \vv_{v,e}) \bigmid \text{$(u,a)\in V_1$, $(v, \vv_{v,e})\in V_4$,
     $uv\in E(G)$, and $a\in e$}\big\}, \\
    E_{22} & = \big\{(u,u)(v,v) \bigmid \text{$(u,u), (v,v)\in V_2$ and $uv\in E(G)$}\big\}, \\
    E_{23} & = \big\{(u,u)(v, \vv_{v,a}) \bigmid \text{$(u,u)\in V_2$, $(v,\vv_{v,a})\in V_3$,
      and $uv\in E(G)$}\big\}, \\
    E_{24} & = \big\{(u,u)(v, \vv_{v,e}) \bigmid \text{$(u,u)\in V_2$, $(v,\vv_{v,e})\in V_4$,
     and $uv\in E(G)$}\big\}, \\
    E_{33} & = \big\{(u,\vv_{u,a})(v, \vv_{v,a}) \bigmid \text{$(u,\vv_{u,a}), (v,\vv_{v,a})\in V_3$
      and $uv\in E(G)$}\big\}, \\
    E_{34} & = \emptyset, \\
    E_{44} & = \big\{(u,\vv_{u,e})(v, \vv_{v,e}) \bigmid \text{$(u,\vv_{u,e}), (v,\vv_{v,e})\in V_4$
      and $uv\in E(G)$}\big\}.
    \end{align*}
\end{enumerate}
\end{definition}

\begin{example}\label{exam:gridproduct}
Let $H$ be the graph in Figure~\ref{fig:gridproduct}~(a). Moreover, we color
every $a_i$ with $\chi(a_i):=a$, every $b_i$ with $\chi(b_i):=b$, and so on.
We consider the grid $\str G_{7,8}$ and the skeleton $\mathcal S= (F,D)$
defined in Example~\ref{exam:gridskeleton}. Then
Figure~\ref{fig:gridproduct}~(b) is the product $P= P(G, \mathcal S, H,
\chi)$. In particular, $V_1$ is the set of white vertices, $V_2$ is the set
of black vertices, $V_3$ is the set of gray vertices, and $V_4$ is the set
of remaining light gray vertices. To make the picture less cluttered, we
label the vertex $(a,a_i)$ by $a_i$ etc. in $P$.
\begin{figure}%[h!]
\centering
\begin{tikzpicture}[scale=1]
 {\scriptsize

% the graph H

\begin{scope}
 \draw (3,3)--(5,3);
 \draw (3,5)--(5,5);

 \draw (3,3)--(3,5);
 \draw (5,3)--(5,5);

 \draw (3,3) node[circle,draw,fill=white]{$a_1$};
 \draw (5,3) node[circle,draw,fill=white]{$b_1$};

 \draw (3,5) node[circle,draw,fill=white]{$c_1$};
 \draw (5,5) node[circle,draw,fill=white]{$d_1$};

 \draw (3.8,3.8)--(5.8,3.8);
 \draw (3.8,5.8)--(5.8,5.8);

 \draw (3.8,3.8)--(3.8,5.8);
 \draw (5.8,3.8)--(5.8,5.8);

 \draw (3.8,3.8) node[circle,draw,fill=white]{$a_2$};
 \draw (5.8,3.8) node[circle,draw,fill=white]{$b_2$};

 \draw (3.8,5.8) node[circle,draw,fill=white]{$c_2$};
 \draw (5.8,5.8) node[circle,draw,fill=white]{$d_2$};

 \draw (4,0) node{{\normalsize (a)}};

\end{scope}

% the grid product

\begin{scope}[xshift=8cm,yshift=0cm]

 \draw (1,1)--(7,1);

% \draw (1,2)--(7,2);
 \draw (1,2)--(2,2);
 \draw (2,2)--(3.2,2.2)--(4,2)--(5.2,2.2)--(6,2);
 \draw (2,2)--(2.8,1.8)--(4,2)--(4.8,1.8)--(6,2);
 \draw (6,2)--(7,2);

% \draw (1,3)--(7,3);
 \draw (1,3)--(2,3)--(3.3,3.3)--(5.3,3.3)--(6,3)--(7,3);
 \draw (1,3)--(2,3)--(2.7,2.7)--(4.7,2.7)--(6,3)--(7,3);
% \draw (6,3)--(7,3);

% \draw (1,4)--(7,4);
 \draw (1,4)--(2,4);
 \draw (2,4)--(3.3,4.3)--(4,4)--(5.3,4.3)--(6,4);
 \draw (2,4)--(2.7,3.7)--(4,4)--(4.7,3.7)--(6,4);
 \draw (6,4)--(7,4);

% \draw (1,5)--(7,5);
 \draw (1,5)--(2,5)--(3.3,5.3)--(5.3,5.3)--(6,5)--(7,5);
 \draw (1,5)--(2,5)--(2.7,4.7)--(4.7,4.7)--(6,5)--(7,5);
% \draw (6,5)--(7,5);

% \draw (1,6)--(7,6);
 \draw (1,6)--(2,6);
 \draw (2,6)--(3.2,6.2)--(4,6)--(5.2,6.2)--(6,6);
 \draw (2,6)--(2.8,5.8)--(4,6)--(4.8,5.8)--(6,6);
 \draw (6,6)--(7,6);

 \draw (1,7)--(7,7);
 \draw (1,8)--(7,8);

 \draw (1,1)--(1,8);

 \draw (2,1)--(2,8);

 \draw (3,1)--(3.2,2.2)--(3.3,3.3)--(3.3,5.3)--(3.2,6.2)--(3,7);
 \draw (3,1)--(2.8,1.8)--(2.7,2.7)--(2.7,4.7)--(2.8,5.8)--(3,7);
 \draw (3,7)--(3,8);

% \draw (13,1)--(13,8);
 \draw (4,1)--(4,2);
 \draw (4,2)--(4.3,3.3)--(4,4)--(4.3,5.3)--(4,6);
 \draw (4,2)--(3.7,2.7)--(4,4)--(3.7,4.7)--(4,6);
 \draw (4,6)--(4,8);

% \draw (5,1)--(5,8);
 \draw (5,1)--(5.2,2.2)--(5.3,3.3)--(5.3,5.3)--(5.2,6.2)--(5,7);
 \draw (5,1)--(4.8,1.8)--(4.7,2.7)--(4.7,4.7)--(4.8,5.8)--(5,7);
 \draw (5,7)--(5,8);

 \draw (6,1)--(6,8);

 \draw (7,1)--(7,8);

% the set $V_2$

 \draw (1,1) node[circle,draw,fill=Black]{};
 \draw (1,2) node[circle,draw,fill=Black]{};
 \draw (1,3) node[circle,draw,fill=Black]{};
 \draw (1,4) node[circle,draw,fill=Black]{};
 \draw (1,5) node[circle,draw,fill=Black]{};
 \draw (1,6) node[circle,draw,fill=Black]{};
 \draw (1,7) node[circle,draw,fill=Black]{};
 \draw (1,8) node[circle,draw,fill=Black]{};

 \draw (2,1) node[circle,draw,fill=Black]{};
 \draw (2,2) node[circle,draw,fill=Black]{};
 \draw (2,3) node[circle,draw,fill=Black]{};
 \draw (2,4) node[circle,draw,fill=Black]{};
 \draw (2,5) node[circle,draw,fill=Black]{};
 \draw (2,6) node[circle,draw,fill=Black]{};
 \draw (2,7) node[circle,draw,fill=Black]{};
 \draw (2,8) node[circle,draw,fill=Black]{};

 \draw (6,1) node[circle,draw,fill=Black]{};
 \draw (6,2) node[circle,draw,fill=Black]{};
 \draw (6,3) node[circle,draw,fill=Black]{};
 \draw (6,4) node[circle,draw,fill=Black]{};
 \draw (6,5) node[circle,draw,fill=Black]{};
 \draw (6,6) node[circle,draw,fill=Black]{};
 \draw (6,7) node[circle,draw,fill=Black]{};
 \draw (6,8) node[circle,draw,fill=Black]{};

 \draw (7,1) node[circle,draw,fill=Black]{};
 \draw (7,2) node[circle,draw,fill=Black]{};
 \draw (7,3) node[circle,draw,fill=Black]{};
 \draw (7,4) node[circle,draw,fill=Black]{};
 \draw (7,5) node[circle,draw,fill=Black]{};
 \draw (7,6) node[circle,draw,fill=Black]{};
 \draw (7,7) node[circle,draw,fill=Black]{};
 \draw (7,8) node[circle,draw,fill=Black]{};

 \draw (3,1) node[circle,draw,fill=Black]{};
 \draw (4,1) node[circle,draw,fill=Black]{};
 \draw (5,1) node[circle,draw,fill=Black]{};

% \draw (3,2) node[circle,draw,fill=Black]{};
 \draw (4,2) node[circle,draw,fill=Black]{};
% \draw (5,2) node[circle,draw,fill=Black]{};

% \draw (3,6) node[circle,draw,fill=Black]{};
 \draw (4,6) node[circle,draw,fill=Black]{};
% \draw (5,6) node[circle,draw,fill=Black]{};

 \draw (3,7) node[circle,draw,fill=Black]{};
 \draw (4,7) node[circle,draw,fill=Black]{};
 \draw (5,7) node[circle,draw,fill=Black]{};

 \draw (3,8) node[circle,draw,fill=Black]{};
 \draw (4,8) node[circle,draw,fill=Black]{};
 \draw (5,8) node[circle,draw,fill=Black]{};

 \draw (4,4) node[circle,draw,fill=Black]{};

% the set $V_3$

% \draw (2.2,3.2) node[circle,draw,fill=gray]{};
% \draw (1.8,2.8) node[circle,draw,fill=gray]{};
%
% \draw (2.2,5.2) node[circle,draw,fill=gray]{};
% \draw (1.8,4.8) node[circle,draw,fill=gray]{};

% \draw (6.2,3.2) node[circle,draw,fill=gray]{};
% \draw (5.8,2.8) node[circle,draw,fill=gray]{};
%
% \draw (6.2,5.2) node[circle,draw,fill=gray]{};
% \draw (5.8,4.8) node[circle,draw,fill=gray]{};

 \draw (2.8,1.8) node[circle,draw,fill=gray]{};
 \draw (3.2,2.2) node[circle,draw,fill=gray]{};

 \draw (4.8,1.8) node[circle,draw,fill=gray]{};
 \draw (5.2,2.2) node[circle,draw,fill=gray]{};

 \draw (2.8,5.8) node[circle,draw,fill=gray]{};
 \draw (3.2,6.2) node[circle,draw,fill=gray]{};

 \draw (4.8,5.8) node[circle,draw,fill=gray]{};
 \draw (5.2,6.2) node[circle,draw,fill=gray]{};

% the set $V_4$

 \draw (3.7,2.7) node[circle,draw,fill=lightgray]{};
 \draw (4.3,3.3) node[circle,draw,fill=lightgray]{};

 \draw (3.7,4.7) node[circle,draw,fill=lightgray]{};
 \draw (4.3,5.3) node[circle,draw,fill=lightgray]{};

 \draw (2.7,3.7) node[circle,draw,fill=lightgray]{};
 \draw (3.3,4.3) node[circle,draw,fill=lightgray]{};

 \draw (4.7,3.7) node[circle,draw,fill=lightgray]{};
 \draw (5.3,4.3) node[circle,draw,fill=lightgray]{};

% the set $V_1$

 \draw (3.3,3.3) node[circle,draw,fill=white]{$a_1$};
 \draw (2.7,2.7) node[circle,draw,fill=white]{$a_2$};

 \draw (5.3,3.3) node[circle,draw,fill=white]{$b_1$};
 \draw (4.7,2.7) node[circle,draw,fill=white]{$b_1$};

 \draw (3.3,5.3) node[circle,draw,fill=white]{$c_1$};
 \draw (2.7,4.7) node[circle,draw,fill=white]{$c_2$};

 \draw (5.3,5.3) node[circle,draw,fill=white]{$d_1$};
 \draw (4.7,4.7) node[circle,draw,fill=white]{$d_2$};

 \draw (4,0) node{{\normalsize (b)}};

\end{scope}

 }

\end{tikzpicture}
\caption{}\label{fig:gridproduct}
\end{figure}
\end{example}

In Definition~\ref{def:product} the reader might notice that in each pair
$(u,a)$, $(u,u)$,  $(u, \vv_{u,a})$, or $(u, \vv_{u,e})$ the first
coordinate is uniquely determined by the second coordinate. Thus:

\begin{lemma}\label{lem:injective}
Let $h: V(G)\to V(P)$ be injective, e.g., $h$ is an embedding from $G$ to
$P$. Then the mapping $\pi_2\circ h$ is injective, too. Here $\pi_2(u,z)= z$
for every $(u,z)\in V(P)$ is the projection on the second coordinate.
\end{lemma}

\medskip
\begin{lemma}\label{lem:pi1hom}
$\pi_1$ is homomorphism from $P$ to $G$, where $\pi_1(u,z)= u$ for every
$(u,z)\in V(P)$ is the projection on the first coordinate.
\end{lemma}

\begin{proof} Observe that by (P2) for every edge $(u,w)(v,z)\in E(P)$ we have
$uv\in E(G)$.
\end{proof}

\begin{lemma}\label{lem:productframe}
Let $h$ be a homomorphism from $G$ to $P$. Then the mapping $\pi_1\circ h$
is an endomorphism of $G$. Moreover, if
\[
\big\{(u,u) \bigmid u\in F\big\}\subseteq h(V(G)),
\]
then $\pi_1\circ h$ is an automorphism of $G$.
\end{lemma}

\begin{proof}
By Lemma~\ref{lem:pi1hom} $\pi_1\circ h$ is an endomorphism of $G$. If
$\big\{(u,u) \bigmid u\in F\big\}\subseteq h(V(G))$, then $F\subseteq
\pi_1\circ h(V(G))$. Since $F$ is a frame, $\pi_1\circ h$ has to be
surjective.
\end{proof}

\begin{lemma}\label{lem:embedproj}
Let $h$ be a homomorphism from $G$ to $P$ such that $\pi_1\circ h$ is an
automorphism of\/ $G$. Then there is a homomorphism $\bar h$ from
$\big(G\setminus F)/D$ to $H$ such that $\chi(\bar h(v))= v$ for every $v\in
V(G)\setminus (F\cup D)$. Note this implies that $\bar h$ is an embedding.
\end{lemma}

\begin{proof}
Let $\rho:= \pi_1\circ h$. By assumption $\rho$ is an automorphism of $G$,
hence so is $\rho^{-1}$. Thus $h\circ \rho^{-1}$ is a homomorphism from $G$
to $P$ with
\begin{equation*}%\label{eq:id}
\pi_1\circ (h\circ \rho^{-1})= (\pi_1\circ h)\circ\rho^{-1}=(\pi_1\circ h)\circ (\pi_1\circ h)^{-1}=
\id.
\end{equation*}
Hence for every $u\in V(G)$ there is a $w$ such that
\begin{equation}\label{eq:h}
h\circ \rho^{-1}(u) = (u, w)
\end{equation}
Let
\[
\bar h:= \pi_2\circ h\circ \rho^{-1}.
\]
We claim that $\bar h$ restricted to $V(G)\setminus (F\cup D)$ is the
desired homomorphism from $(G\setminus F)/D$ to $H$.

\medskip
Let $u\in V(G)\setminus (F\cup D)$. By the definition of $V(P)$ in (P1)
and~\eqref{eq:h} we have $\bar h(u)\in V_1$, and thus by the definition of
$V_1$,
\[
\bar h(u)\in V(H) \ \text{\; with \; $\chi(\bar h(u))= u$}.
\]
Next, let $uv\in E\big((G\setminus F)/D\big)$. We have to show $\bar h(u)
\bar h(v)\in E(H)$. By the definition of $(G\setminus F)/D$ there is a path
\[
u=v_1 \to v_2 \to \cdots \to v_k= v
\]
in $G\setminus F$ with $k\ge 2$ and all $v_i\in D$ for $1< i < k$. If $k=2$,
then $uv\in E(G)$. We can conclude
\[
\left\{\big(u, \bar h(u)\big), \big(v, \bar h(v)\big)\right\}=\left\{h\circ \rho^{-1}(u), h\circ \rho^{-1}(v)\right\}
 \in E(P),
\]
because $h\circ \rho^{-1}$ is a homomorphism. Then $\big\{\bar h(u), \bar
h(v)\big\}\in E(H)$ follows directly from the definition of $E_{11}$ in
(P2). So assume $k> 2$. Again by~\eqref{eq:h} and (P1) for some pairwise
distinct $a,b\in V(H)$ and $w_2, \ldots, w_{k-1}$
\begin{align*}
h\circ \rho^{-1}(u) &= (u,a), \\
h\circ \rho^{-1}(v_2) &= (v_2, w_2), \\
 & \ \ \vdots \\
h\circ \rho^{-1}(v_{k-1}) &= (v_{k-1}, w_{k-1}), \\
h\circ \rho^{-1}(v) &= (v, b).
\end{align*}
As every $v_i$ is associated with $\{u,v\}$, there are $e_2, \ldots,
e_{k-1}\in E(H)$ with $w_i= \vv_{v_i, e_i}$ by the definition of $V_4$ in
(P1). Since $h\circ \rho^{-1}$ is a homomorphism from $G$ to $P$,
\begin{align*}
\big\{(u,a), (v_2,\vv_{v_2, e_2})\big\}&\in E(P), \\
\big\{(v_2, \vv_{v_2, e_2}),(v_3, \vv_{v_3, e_3})\big\}&\in E(P), \\
  &\vdots \\
\big\{(v_{k-2}, \vv_{v_{k-2}, e_{k-2}}), (v_{k-1}, \vv_{v_{k-1}, e_{k-1}})\big\}&\in E(P), \\
\big\{(v_{k-1}, \vv_{v_{k-1}, e_{k-1}}), (v,b)\big\}&\in E(P).
\end{align*}
Then by the definition of $E_{44}$ in (P2), we conclude $e_2= \cdots
=e_{k-1}$. Finally, the definition of $E_{14}$ implies that $e_2= \{a,b\}$,
i.e., $\big\{\bar h(u), \bar h(v)\big\}\in E(H)$.
\end{proof}

\begin{lemma}\label{lem:embedlift}
If there is an embedding $\bar h$ from $\big(G\setminus F)/D$ to $H$ with
$\chi(\bar h(v))= v$ for every $v\in V(G)\setminus (F\cup D)$, then there is
an embedding from $G$ to $P$.
\end{lemma}

\begin{proof}
We define a mapping $h: V(G)\to V(P)$ and show that it is an embedding.
\begin{itemize}
\item For $u\in V(G)\setminus (F\cup D)$ let $h(u):= \big(u, \bar
    h(u)\big)$, which is well defined by $\chi(\bar h(u))= u$.

\item For $u\in F$ let $h(u):= (u,u)$.

\item For $u\in D$ without being associated with any vertex or edge in
    $(G\setminus F)/D$ let $h(u):= (u,u)$.

\item Let $u\in D$ be associated with a (unique) $v \in V\big((G\setminus
    F)/D\big)$. We set $h(u):= \big(u, \vv_{u, \bar h(v)}\big)$.

\item Let $u\in D$ be associated with a (unique) %\yrand{This is the only place we need $D$ consist of vertices of degree $\le 2$.}
    $vw \in E\big((G\setminus F)/D\big)$. We set $h(u):= \big(u, \vv_{u,
    \bar h(v) \bar h(w)}\big)$.
\end{itemize}
The injectivity of $h$ is trivial. To see that it is a homomorphism, let
$uv\in E(G)$ and we need to establish $h(u)h(v)\in E(P)$.
\begin{itemize}
\item Assume $u,v\in V(G)\setminus (F\cup D)$. Then $uv\in E(G)$
  implies $uv\in E\big((G\setminus F)/D\big)$, and as $\bar h$ is a
    homomorphism from $\big(G\setminus F\big)/D$ to $H$, it follows that
    $\bar h(u)\bar h(v)\in E(H)$. So by the definition of $E_{11}$ in (P1)
    we conclude $\big(u, \bar h(u)\big) \big(v, \bar h(v)\big)\in E(P)$.

\item Let $u\in V(G)\setminus (F\cup D)$ and $v\in D$. Furthermore, assume
    that $v$ is associated with an edge $wz\in E\big((G\setminus
    F)/D\big)$.
    %let $wz\in E\big((G\setminus F)/D\big)$ be the edge with which $v$ is
    %associated.
    Hence, $h(u)= \big(u, \bar h(u)\big)$ and $h(v)= \big(v, \vv_{v, \bar
    h(w)\bar h(z)}\big)$. Recall $uv\in E(G)$, therefore $u= w$ or $u= z$.
    Then, $\bar h(u)\in \big\{\bar h(w), \bar h(z)\big\}$, and the
    definition of $E_{14}$ in (P2) implies that $h(u) h(v)\in E(P)$.

\item Assume both $u,v\in D$ and they are associated with some edges
    $e_1,e_2\in E\big((G\setminus F)/D\big)$. Then $e_1=e_2$ by $uv\in
    E(G)$, and $h(u)h(v)\in E(P)$ follows from the definition of $E_{44}$
    in (P2).

\item All the remaining cases are similar and easy.
\end{itemize}
\end{proof}

\begin{definition}\label{def:rigidskeleton}
A skeleton $\mathcal S= (F, D)$ is \emph{rigid} if for every graph $H$,
every $\chi: V(H) \to V(G)\setminus (F\cup D)$, and every \emph{embedding}
$h$ from $G$ to $P= P(G, \mathcal S, H, \chi)$, it holds that $\big\{(u,u)
\bigmid u\in F\big\}\subseteq h(V(G))$.
\end{definition}

\begin{proposition}\label{prop:rigidcomputable}
There is an algorithm which lists all rigid skeletons of an input graph $G$.
\end{proposition}

\begin{proof}
Let $G$ be a graph and $F, D\subseteq V(G)$. Clearly it is
decidable whether $\mathcal S= (F,D)$ is a skeleton by
Definition~\ref{def:skeleton}. Moreover, we observe that $\mathcal S$ is
\emph{not} a rigid skeleton if and only if there is graph $H$, a mapping
$\chi: V(H)\to V(G)\setminus (F\cup D)$, and an embedding $h$ from $G$ to
$P= P(G, \mathcal S, H, \chi)$ such that
\begin{equation}\label{eq:notrigid}
\big\{(u,u) \bigmid u\in F\big\}\not\subseteq h(V(G)).
\end{equation}
We define a set
\begin{align*}
X= & \ \big\{a\in V(H) \bigmid \text{$(u,a)\in h(G)$ for some $u\in V(G)\setminus (F\cup D)$}\big\} \\
 & \qquad \cup \left\{ a\in V(H) \bigmid \text{$(u, \vv_{u, a})\in h(G)$ for some $u\in D$}\right\} \\
 & \qquad \cup \left\{ a,b\in V(H) \bigmid \text{$(u, \vv_{u, ab})\in h(G)$ for some $u\in D$}\right\}.
\end{align*}
It is routine to verify that $h$ is an embedding from $G$ to $P'= P\big(G,
\mathcal S, H[X], \chi_{\upharpoonright X}\big)$ such that
\eqref{eq:notrigid} also holds. Hence, the induced subgraph $H[X]$ with the
coloring $\chi_{\upharpoonright X}$ also witnesses that $\mathcal S$ is not
rigid. Observe that $|X| \le 2|V(G)|$.

\medskip
Therefore, to list all the rigid skeletons of $G$, we enumerate all pairs
$\mathcal S= (F,D)$,
\begin{itemize}
\item check whether $\mathcal S$ is a skeleton,

\item and if so, then check whether it is rigid by going through all
    graphs on the vertex set $[n]$ with $n\le 2|V(G)|$.
\end{itemize}
\end{proof}

\begin{definition}
A class $\cls K$ of graphs is \emph{rich} if for every $k\in \mathbb N$
there is a graph $G\in \cls K$ such that $G$ has a rigid skeleton $(F,D)$
with
\begin{equation}\label{eq:tw}
\tw\big((G\setminus F)/D\big)\ge k.
\end{equation}
\end{definition}

\begin{theorem}\label{thm:richW1}
Let $\cls K$ be a recursively enumerable and rich class of graphs. Then
$\pEmb(\cls K)$ is hard for \W 1.
\end{theorem}

\begin{proof}
We define a sequence of graphs $G_1, G_2, \ldots$ and sets $F_i,D_i\subseteq
V(G_i)$ as follows. For every $i\in \mathbb N$ we enumerate graphs $G$ in
the class $\cls K$ one by one. For every $G$ we list all the rigid skeletons
$(F, D)$ of $G$ by Proposition~\ref{prop:rigidcomputable}. Then we check
whether there is such a rigid skeleton $(F, D)$ satisfying~\eqref{eq:tw}. If
so, we let $G_i:= G$, $(F_i,D_i):= (F,D)$, and define
\[
G^*_i:= (G_i\setminus F_i)/D_i.
\]
By our assumption, $G_i$ will be found eventually, and $G^*_i$ is well
defined and computable from $G_i$. It follows that the class
\[
\cls K^*:= \big\{G^*_i \bigmid i\in \mathbb N\big\}
\]
is recursively enumerable and has unbounded treewidth.

By Lemma~\ref{lem:colembedW1}, we conclude that $\pColEmb(\cls K^*)$ is
\W 1-hard. Hence it suffices to give an \fpt-reduction from $\pColEmb(\cls
K^*)$ to $\pEmb(\cls K)$. Let $G^*_i\in \cls K^*$. Thus $G^*_i=
(G_i\setminus F_i)/D_i$ for the rigid skeleton $\mathcal S_i= (F_i,D_i)$.
Then for every graph $H$ and $\chi: V(H)\mapsto V(G^*_i)$ we claim that
\begin{align*}
\text{there is an embedding $h$ from} & \ \text{$G^*_i$ to $H$ with $\chi(h(v))= v$ for every $v\in V(G^*_i)$} \\
 & \iff \text{there is an embedding from $G_i$ to $P\big(G_i, \mathcal S_i, H, \chi\big)$}.
\end{align*}
The direction from left to right is by Lemma~\ref{lem:embedlift}. The other
direction follows from the rigidity of $\mathcal S_i$,
Lemma~\ref{lem:productframe}, and Lemma~\ref{lem:embedproj}.
\end{proof}

%In Section~\ref{sec:richnessgridwall} we will prove the richness of grids
%and walls. More precisely:
%
%\begin{proposition}\label{prop:gridwallrich}
%Let $\cls K$ be a class of graphs.
%\begin{enumerate}
%\item[(i)] If every $k\in \mathbb N$ there exists a grid $\str G_{s,t}\in
%    \cls K$ with $\min \{s,t\}\ge k$. Then $\cls K$ is rich.
%
%\item[(ii)] If every $k\in \mathbb N$ there exists a wall $\str W_{s,t}\in
%    \cls K$ with $\min \{s,t\}\ge k$. Then $\cls K$ is rich.
%\end{enumerate}
%\end{proposition}
%
%Now the following more general version of Theorem~\ref{thm:main} is an
%immediate consequence of Theorem~\ref{thm:richW1} and
%Proposition~\ref{prop:gridwallrich}.
%\begin{theorem}
%Let $\cls K$ be a recursively enumerable class of graphs. Then $\pEmb(\cls
%K)$ is \W 1-hard if one of the following conditions is satisfied.
%\begin{enumerate}
%\item For every $k\in \mathbb N$ there exists a grid $\str G_{k_1,k_2}\in
%    \cls K$ with $\min\{k_1, k_2\}\ge k$.
%
%\item For every $k\in \mathbb N$ there exists a wall $\str W_{k_1,k_2}\in
%    \cls K$ with $\min\{k_1, k_2\}\ge k$.
%\end{enumerate}
%\end{theorem}

\section{Grids and Walls}\label{sec:richnessgridwall}
In this section we show that the classes of grids and walls are rich. More
precisely:

\begin{proposition}\label{prop:gridwallrich}
Let $\cls K$ be a class of graphs.
\begin{enumerate}
\item[(i)] If every $k\in \mathbb N$ there exists a grid $\str G_{s,t}\in
    \cls K$ with $\min \{s,t\}\ge k$. Then $\cls K$ is rich.

\item[(ii)] If every $k\in \mathbb N$ there exists a wall $\str W_{s,t}\in
    \cls K$ with $\min \{s,t\}\ge k$. Then $\cls K$ is rich.
\end{enumerate}
\end{proposition}

Now the following more general version of Theorem~\ref{thm:main} is an
immediate consequence of Theorem~\ref{thm:richW1} and
Proposition~\ref{prop:gridwallrich}.
\begin{theorem}
Let $\cls K$ be a recursively enumerable class of graphs. Then $\pEmb(\cls
K)$ is \W 1-hard if one of the following conditions is satisfied.
\begin{enumerate}
\item For every $k\in \mathbb N$ there exists a grid $\str G_{k_1,k_2}\in
    \cls K$ with $\min\{k_1, k_2\}\ge k$.

\item For every $k\in \mathbb N$ there exists a wall $\str W_{k_1,k_2}\in
    \cls K$ with $\min\{k_1, k_2\}\ge k$.
\end{enumerate}
\end{theorem}

%
%Grids will be our main focus. The proof for walls is easier.
%
%\begin{definition}\label{def:gridwall}
%Let $s,t\in \mathbb N$. A \emph{$(s\times t)$-grid $\str G_{s, t}$} has
%\begin{eqnarray*}
%V(\str G_{s, t})= [s]\times [t] & \text{and} &
%E(\str G_{s, t})= \big\{(i,j)(i',j')\bigmid |i-i'|+ |j-j'|=1\big\}.
%\end{eqnarray*}
%And the \emph{wall $\str W_{s, t}$} of width $s$ and height $t$ is defined
%by
%\begin{align*}
%V(\str W_{s, t}) & = \big\{v_{i,j} \bigmid \text{$i\in [s+1]$ and $j\in [t]$}\big\}
% \cup \big\{v_{i,t+1} \bigmid \text{$i\in [s+1]$ and odd $t$}\big\} \\
% & \quad \cup \big\{u_{i,j} \bigmid \text{$i\in [s+1]$ and $2\le j\le t$]}\big\}
%  \cup \big\{u_{i,t+1} \bigmid \text{$i\in [s+1]$ and even $t$}\big\}, \\
%E(\str W_{s, t}) & = \big\{v_{i,1}v_{i+1,1}\bigmid i\in [s]\big\} \\
% & \quad \cup \big\{v_{i,t+1}v_{i+1,t+1}\bigmid \text{$i\in [s]$ and odd $t$}\big\} \\
% & \quad \cup \big\{u_{i,t+1}u_{i+1,t+1}\bigmid \text{$i\in [s]$ and even $t$}\big\} \\
% & \quad \cup \big\{v_{i,j}u_{i,j}, u_{i,j}v_{i+1,j} \bigmid \text{$i\in [s]$ and $2\le j\le t$}\big\} \\
% & \quad \cup \big\{v_{i,j}v_{i,j+1} \bigmid \text{$i\in [s+1]$ and odd $j\in [t]$}\big\} \\
% & \quad \cup \big\{u_{i,j}u_{i,j+1} \bigmid \text{$i\in [s+1]$ and even $j\in [t]$}\big\}.
%\end{align*}
%\end{definition}
%
\subsection{The richness of grids}\label{subsec:grids}

\begin{lemma}\label{lem:gridframe}
Let $s,t\in \mathbb N$. Then the set of the four corner vertices
\[
F:= \big\{(1,1), (s,1), (1,t), (s,t) \big\}
\]
is a frame for $\str G_{s,t}$.
\end{lemma}

\begin{proof}
Let $h$ be an endomorphism of $\str G_{s,t}$ with $F\subseteq h(\str
G_{s,t})$.

\medskip
\noindent \textit{Claim 1.} $h^{-1}(F)= F$. More precisely,
\begin{eqnarray*}
h^{-1}\big(\{(1,1), (s,t)\}\big)= \{(1,1), (s,t)\}
 & \text{and} &
h^{-1}\big(\{(1,t), (s,1)\}\big)= \{(1,t), (s,1)\},
\end{eqnarray*}
or
\begin{eqnarray*}
h^{-1}\big(\{(1,1), (s,t)\}\big)= \{(1,t), (s,1)\}
 & \text{and} &
h^{-1}\big(\{(1,t), (s,1)\}\big)= \{(1,1), (s,t)\}.
\end{eqnarray*}

\medskip
\noindent \textit{Proof of the claim.} As $F\subseteq h(\str G_{s,t})$,
\[
\{u,v\}:= h^{-1}\big(\{(1,1), (s,t)\}\big).
\]
is well defined. Clearly
\[
\dist^{\str G_{s,t}}(u,v)
\ge \dist^{\str G_{s,t}}\big(h(u),h(v)\big)
= \dist^{\str G_{s,t}}\big((1,1), (s,t)\big) = s+t-2,
\]
since $h$ is a homomorphism. But this is only possible if
\begin{eqnarray*}
\{u,v\}= \big\{(1,1), (s,t)\big\} & \text{or} & \{u,v\}= \big\{(1,t),
(s,1)\big\}.
\end{eqnarray*}
By the same reasoning,
\begin{eqnarray*}
h^{-1}\big(\{(1,t), (s,1)\}\big) = \big\{(1,t), (s,1)\big\} & \text{or} &
h^{-1}\big(\{(1,t), (s,1)\}\big) = \big\{(1,1), (s,t)\big\}.
\end{eqnarray*}
The claim then follows easily. \hfill$\dashv$

\medskip
So without loss of generality we can assume that $h(u)= u$ for every $u\in
F$.

\medskip
\noindent \textit{Claim 2.} Let $u,v\in V(\str G_{s,t})$. If for every $w\in
F$ we have
\[
\dist^{\str G_{s,t}}(u,w)\ge \dist^{\str G_{s,t}}(v,w),
\]
then $u=v$.

\medskip
\noindent \textit{Proof of the claim.} Routine. \hfill$\dashv$

\medskip
Now we are ready to show that $h(u)= u$ for every $u\in V(\str G_{s,t})$.
Note that for every $w\in F$
\begin{align*}
\dist^{\str G_{s,t}}(u,w)
 & \ge \dist^{\str G_{s,t}}(h(u),h(w)) & \text{($h$ is a homomorphism)}\\
 & = \dist^{\str G_{s,t}}(h(u),w) & \text{(our assumption $h(w)= w$)}.
\end{align*}
Claim~2 implies that $h(u)=u$.
\end{proof}

\begin{definition}\label{defn:gridskeleton}
Let $s,t\in \mathbb N$ with $s\ge 5$ and $t\ge 6$. Set
\begin{eqnarray}\label{eq:k1k2}
k_1:= %\max \big\{k\in \mathbb N \bigmid 2k< s\big\}=
  \floor{\frac{s-1}{2}}
 & \text{and} &
k_2:= %\max \big\{k\in \mathbb N \bigmid 2k\le t-2\big\} =
  \floor{\frac{t-2}{2}}.
\end{eqnarray}
Then we define $\mathcal S_{s,t}:= (F,D)$ with
\begin{align*}
F &= \big\{(i,j)\bigmid \text{\big($i\in [2]$ or\; $2k_1\le i \le s$\big)
 or \big($j=1$ or\; $2k_2< j \le t$\big)}\big\} \\
 & \hspace{5.8cm}\cup \big\{(2i,2j)\bigmid \text{$i\in [k_1]$ and $j\in [k_2]$}\big\}, \\
D &= \big\{(2i+1,2j) \bigmid \text{$i\in [k_1-1]$ and $j\in [k_2]$}\big\} \\
 & \hspace{3.9cm}
 \cup \big\{(2i,2j+1) \bigmid \text{$1< i< k_1$ and $j\in [k_2-1]$}\big\}.
\end{align*}
See Figure~\ref{fig:gridskeleton} for $\mathcal S_{7,8}$.
%\mrand{A picture would help}
\end{definition}

In the following we fix some $s,t\in \mathbb N$ with $s\ge 5$ and $t\ge 6$.

\begin{lemma}\label{lem:gridskeleton}
$\mathcal S_{s,t}$ is a skeleton of the grid $\str G_{s, t}$.
\end{lemma}

\begin{proof}
The condition (S1) in Definition~\ref{def:skeleton} follows from
Lemma~\ref{lem:gridframe} and Remark~\ref{rem:frame}~(3). The conditions
(S2) and (S3) are immediate.
\end{proof}

\begin{remark}\label{rem:gridskeletion}
The frame in Definition~\ref{defn:gridskeleton} contains much more vertices
than what is required in Lemma~\ref{lem:gridframe}. It will become evident
for our proof that $F$ needs to contain top two rows, leftmost two columns,
and rightmost two columns in the grid.
\end{remark}

The next result is routine.
\begin{lemma}\label{lem:gridskeletontw}
$(\str G_{s, t}\setminus F)/D$ is a $\big((\floor{(s-1)/2}-1)\times
(\floor{(t-2)/2}-1)\big)$-grid, where $(F,D)= \mathcal S_{s,t}$. Hence
\[
\tw\big((\str G_{s, t}\setminus F)/D\big)=
 \min\left\{\floor{\frac{s-1}{2}}, \floor{\frac{t-2}{2}}\right\}-1.
\]
\end{lemma}

\begin{proposition}\label{prop:gridrigid}
$\mathcal S_{s,t}$ is a rigid skeleton of the grid $\str G_{s, t}$.
\end{proposition}

The proof is very technical. We proceed in several steps to improve the
readability.

Note that the shape of $F$ is slightly different depending on the parity of
$s$ and $t$. Without loss of generality, we assume that $s$ is odd and $t$
is even as in Example~\ref{exam:gridskeleton} and
Example~\ref{exam:gridproduct}. Therefore in~\eqref{eq:k1k2} we have $k_1=
(s-1)/2$ and $k_2= t/2-1$.

We fix a graph $H$, a mapping $\chi: V(H)\to V(\str G_{s, t})\setminus
(F\cup D)$, and an embedding $h$ from $\str G_{s, t}$ to $P= P(\str G_{s,
t}, \mathcal S_{s,t}, H, \chi)$. Moreover, let
\[
F_1:= \big\{(2i, 2j) \bigmid \text{$i\in [k_1]$ and $j\in [k_2+1]$}\big\}.
\]

\begin{lemma}\label{lem:C}
Let
\[
C:=\big\{(u,u) \bigmid u\in F_1\big\}
\]
(which is shown in Figure~\ref{fig:gridC} for our running example).
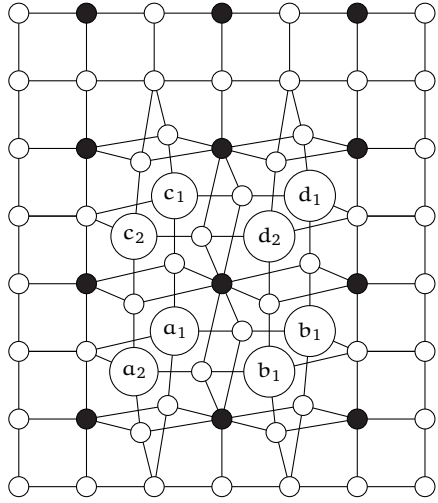
\begin{figure}%[h!]
\centering
\begin{tikzpicture}[scale=.9]
 {\scriptsize

% the grid product

 \draw (1,1)--(7,1);

% \draw (1,2)--(7,2);
 \draw (1,2)--(2,2);
 \draw (2,2)--(3.2,2.2)--(4,2)--(5.2,2.2)--(6,2);
 \draw (2,2)--(2.8,1.8)--(4,2)--(4.8,1.8)--(6,2);
 \draw (6,2)--(7,2);

% \draw (1,3)--(7,3);
 \draw (1,3)--(2,3)--(3.3,3.3)--(5.3,3.3)--(6,3)--(7,3);
 \draw (1,3)--(2,3)--(2.7,2.7)--(4.7,2.7)--(6,3)--(7,3);
% \draw (6,3)--(7,3);

% \draw (1,4)--(7,4);
 \draw (1,4)--(2,4);
 \draw (2,4)--(3.3,4.3)--(4,4)--(5.3,4.3)--(6,4);
 \draw (2,4)--(2.7,3.7)--(4,4)--(4.7,3.7)--(6,4);
 \draw (6,4)--(7,4);

% \draw (1,5)--(7,5);
 \draw (1,5)--(2,5)--(3.3,5.3)--(5.3,5.3)--(6,5)--(7,5);
 \draw (1,5)--(2,5)--(2.7,4.7)--(4.7,4.7)--(6,5)--(7,5);
% \draw (6,5)--(7,5);

% \draw (1,6)--(7,6);
 \draw (1,6)--(2,6);
 \draw (2,6)--(3.2,6.2)--(4,6)--(5.2,6.2)--(6,6);
 \draw (2,6)--(2.8,5.8)--(4,6)--(4.8,5.8)--(6,6);
 \draw (6,6)--(7,6);

 \draw (1,7)--(7,7);
 \draw (1,8)--(7,8);

 \draw (1,1)--(1,8);

 \draw (2,1)--(2,8);

 \draw (3,1)--(3.2,2.2)--(3.3,3.3)--(3.3,5.3)--(3.2,6.2)--(3,7);
 \draw (3,1)--(2.8,1.8)--(2.7,2.7)--(2.7,4.7)--(2.8,5.8)--(3,7);
 \draw (3,7)--(3,8);

 \draw (4,1)--(4,2);
 \draw (4,2)--(4.3,3.3)--(4,4)--(4.3,5.3)--(4,6);
 \draw (4,2)--(3.7,2.7)--(4,4)--(3.7,4.7)--(4,6);
 \draw (4,6)--(4,8);

 \draw (5,1)--(5.2,2.2)--(5.3,3.3)--(5.3,5.3)--(5.2,6.2)--(5,7);
 \draw (5,1)--(4.8,1.8)--(4.7,2.7)--(4.7,4.7)--(4.8,5.8)--(5,7);
 \draw (5,7)--(5,8);

 \draw (6,1)--(6,8);

 \draw (7,1)--(7,8);

% the set $V_2$

 \draw (1,1) node[circle,draw,fill=white]{};
 \draw (1,2) node[circle,draw,fill=white]{};
 \draw (1,3) node[circle,draw,fill=white]{};
 \draw (1,4) node[circle,draw,fill=white]{};
 \draw (1,5) node[circle,draw,fill=white]{};
 \draw (1,6) node[circle,draw,fill=white]{};
 \draw (1,7) node[circle,draw,fill=white]{};
 \draw (1,8) node[circle,draw,fill=white]{};

 \draw (2,1) node[circle,draw,fill=white]{};
 \draw (2,2) node[circle,draw,fill=Black]{};
 \draw (2,3) node[circle,draw,fill=white]{};
 \draw (2,4) node[circle,draw,fill=Black]{};
 \draw (2,5) node[circle,draw,fill=white]{};
 \draw (2,6) node[circle,draw,fill=Black]{};
 \draw (2,7) node[circle,draw,fill=white]{};
 \draw (2,8) node[circle,draw,fill=Black]{};

 \draw (6,1) node[circle,draw,fill=white]{};
 \draw (6,2) node[circle,draw,fill=Black]{};
 \draw (6,3) node[circle,draw,fill=white]{};
 \draw (6,4) node[circle,draw,fill=Black]{};
 \draw (6,5) node[circle,draw,fill=white]{};
 \draw (6,6) node[circle,draw,fill=Black]{};
 \draw (6,7) node[circle,draw,fill=white]{};
 \draw (6,8) node[circle,draw,fill=Black]{};

 \draw (7,1) node[circle,draw,fill=white]{};
 \draw (7,2) node[circle,draw,fill=white]{};
 \draw (7,3) node[circle,draw,fill=white]{};
 \draw (7,4) node[circle,draw,fill=white]{};
 \draw (7,5) node[circle,draw,fill=white]{};
 \draw (7,6) node[circle,draw,fill=white]{};
 \draw (7,7) node[circle,draw,fill=white]{};
 \draw (7,8) node[circle,draw,fill=white]{};

 \draw (3,1) node[circle,draw,fill=white]{};
 \draw (4,1) node[circle,draw,fill=white]{};
 \draw (5,1) node[circle,draw,fill=white]{};

% \draw (3,2) node[circle,draw,fill=white]{};
 \draw (4,2) node[circle,draw,fill=Black]{};
% \draw (5,2) node[circle,draw,fill=white]{};

% \draw (3,6) node[circle,draw,fill=white]{};
 \draw (4,6) node[circle,draw,fill=Black]{};
% \draw (5,6) node[circle,draw,fill=white]{};

 \draw (3,7) node[circle,draw,fill=white]{};
 \draw (4,7) node[circle,draw,fill=white]{};
 \draw (5,7) node[circle,draw,fill=white]{};

 \draw (3,8) node[circle,draw,fill=white]{};
 \draw (4,8) node[circle,draw,fill=Black]{};
 \draw (5,8) node[circle,draw,fill=white]{};

 \draw (4,4) node[circle,draw,fill=Black]{};

% the set $V_3$

 \draw (2.8,1.8) node[circle,draw,fill=white]{};
 \draw (3.2,2.2) node[circle,draw,fill=white]{};

 \draw (4.8,1.8) node[circle,draw,fill=white]{};
 \draw (5.2,2.2) node[circle,draw,fill=white]{};

 \draw (2.8,5.8) node[circle,draw,fill=white]{};
 \draw (3.2,6.2) node[circle,draw,fill=white]{};

 \draw (4.8,5.8) node[circle,draw,fill=white]{};
 \draw (5.2,6.2) node[circle,draw,fill=white]{};

% the set $V_4$

 \draw (3.7,2.7) node[circle,draw,fill=white]{};
 \draw (4.3,3.3) node[circle,draw,fill=white]{};

 \draw (3.7,4.7) node[circle,draw,fill=white]{};
 \draw (4.3,5.3) node[circle,draw,fill=white]{};

 \draw (2.7,3.7) node[circle,draw,fill=white]{};
 \draw (3.3,4.3) node[circle,draw,fill=white]{};

 \draw (4.7,3.7) node[circle,draw,fill=white]{};
 \draw (5.3,4.3) node[circle,draw,fill=white]{};

% the set $V_1$

 \draw (3.3,3.3) node[circle,draw,fill=white]{$a_1$};
 \draw (2.7,2.7) node[circle,draw,fill=white]{$a_2$};

 \draw (5.3,3.3) node[circle,draw,fill=white]{$b_1$};
 \draw (4.7,2.7) node[circle,draw,fill=white]{$b_1$};

 \draw (3.3,5.3) node[circle,draw,fill=white]{$c_1$};
 \draw (2.7,4.7) node[circle,draw,fill=white]{$c_2$};

 \draw (5.3,5.3) node[circle,draw,fill=white]{$d_1$};
 \draw (4.7,4.7) node[circle,draw,fill=white]{$d_2$};

% \draw (13,0) node{{\normalsize (b)}};

 }

\end{tikzpicture}
\caption{The set $C$ consists of those black vertices.}\label{fig:gridC}
\end{figure}
\begin{enumerate}
\item[(C1)] Every $4$-cycle in $P$ must contain a vertex in $C$.

\item[(C2)] Every pair of vertices in $C$ has even distance in $P$.
\end{enumerate}
\end{lemma}

\begin{proof}
Recall that $\pi_1$ is a homomorphism from $P$ to $\str  G_{s,t}$ with
$\pi_1(u,z)=u$ for every $(u,z)\in V(P)$. Obviously, $\pi_1(C)=F_1$. Since
$P$ and $\str G_{s,t}$ have no loops, every homomorphism from $P$ to $\str
G_{s,t}$ preserves the parity of walk length. In particular, if
$(u_1,z_1),(u_2,z_2)\in V(P)$ have odd distance in $P$, then $u_1$ and $u_2$
must have a walk with odd length in $\str G_{s,t}$. It is not hard to see
that (C2) follows from the fact  that every walk for a pair of vertices in
$F_1$ has even length in $\str G_{s,t}$.

To prove (C1), we observe that every $4$-cycle in $\str G_{s,t}$ must
contain a vertex in $F_1$. If a $4$-cycle $X$ in $P$ is mapped to a
$4$-cycle in $\str G_{s,t}$ by $\pi_1$, then $X$ must contain a vertex in
$C$ because $\pi_1^{-1}(F_1)=C$. So it suffices to argue that no $4$-cycle
is mapped to an edge under $\pi_1$. Suppose a $4$-cycle $X$ in $P$ is mapped
to an edge under $\pi_1$. Then there must exist $u,u'\in V(\str G_{s,t})$
such that $V(X)=\{(u,z_1),(u',z_2),(u,z_3),(u',z_4)\}$ and it forms a
$4$-cycle
\[
X: (u, z_1) \to (u', z_2) \to (u, z_3) \to (u', z_4) \to (u, z_1).
\]
Recall that the vertex set of $P$ is partitioned into $4$ subsets $V_1$,
$V_2$, $V_3$, and $V_4$.

\begin{itemize}
\item First, we note that for every vertex $(v,b)\in V_2$, there is no
    vertex $(v,b')\in V(P)$ with $b'\neq b$. Thus $(u,z_1)$ can not be in
    $V_2$, otherwise $z_3=z_1$, contradicting the fact that $X$ is a
    $4$-cycle. By the similar argument, we have $V(X)\cap
    V_2=\varnothing$.

\item If $(u,z_1)\in V_3$, then so is $(u,z_3)$. We first note that
    $(u',z_2),(u',z_4)$ can not be in  $V_4$ because there is no edge
    between $V_3$ and $V_4$.

    If $(u',z_2)$ or $(u',z_4)$ is in $V_3$, then both are in $V_3$. By
    the definition of $E_{3,3}$ this means that $z_1=z_2=z_3=z_4$. It
    follows that $(u,z_1)=(u,z_3)$, which is contradiction.

    So let us assume that $(u',z_2),(u',z_4)\in V_1$.  By the definition
    of $E_{13}$ and $V_3$, we have $z_1=z_3=\vv_{u,z_2}$, which leads to a
    contradiction.

\item If $(u,z_1)\in V_4$, then by similar arguments, we must have
    $(u,z_3)\in V_4$ and $(u',z_2),(u',z_4)\in V_1$.
By the definition of $E_{14}$ and since $z_2\neq z_4$ we have
$z_1=z_3=\vv_{u,\{z_2,z_4\}}$. Again this leads to a contradiction.
\item Finally, assume that $V(X)\subseteq V_1$. Thus $u,u'\in V(\str
    G_{s,t})\setminus(F\cup D)$ and $uu'\in E(\str G_{s,t})$. However it
    is easy to see that the vertices in $V(\str G_{s,t})\setminus(F\cup
    D)$ are mutually nonadjacent. %\qed
\end{itemize}
\end{proof}

We define a further function
\[
\bar h(i,j):= h(i, t+1-j).
\]
for every $i\in [s]$ and $j\in [t]$. The mapping $\bar h$ is also an
embedding from $\str G_{s,t}$ to $P$, since $(i,j) \mapsto (i, t+1-j)$ is an
automorphism of $\str G_{s,t}$.

\begin{lemma}\label{lem:correctcenter}
Either $h(F_1)= C$ or $\bar h(F_1)= C$.
\end{lemma}

\begin{figure}%[h!]
\centering
\begin{tikzpicture}[scale=.8]
 {\scriptsize

% the grid product

 \draw (1,1)--(7,1);
 \draw (1,2)--(7,2);
 \draw (1,3)--(7,3);
 \draw (1,4)--(7,4);
 \draw (1,5)--(7,5);
 \draw (1,6)--(7,6);
 \draw (1,7)--(7,7);
 \draw (1,8)--(7,8);

 \draw (1,1)--(1,8);
 \draw (2,1)--(2,8);
 \draw (3,1)--(3,8);
 \draw (4,1)--(4,8);
 \draw (5,1)--(5,8);
 \draw (6,1)--(6,8);
 \draw (7,1)--(7,8);

% \draw (1,1)--(2,1)--(2,2)--(1,2)--(1,1) node[circle,draw,fill=lightgray]{$d_2$};

  \draw (1.5,1.5) node [minimum size=.8cm,fill=lightgray,draw]{};
  \draw (1.5,3.5) node [minimum size=.8cm,fill=lightgray,draw]{};
  \draw (1.5,5.5) node [minimum size=.8cm,fill=lightgray,draw]{};
  \draw (1.5,7.5) node [minimum size=.8cm,fill=lightgray,draw]{};

  \draw (3.5,1.5) node [minimum size=.8cm,fill=lightgray,draw]{};
  \draw (3.5,3.5) node [minimum size=.8cm,fill=lightgray,draw]{};
  \draw (3.5,5.5) node [minimum size=.8cm,fill=lightgray,draw]{};
  \draw (3.5,7.5) node [minimum size=.8cm,fill=lightgray,draw]{};

  \draw (5.5,1.5) node [minimum size=.8cm,fill=lightgray,draw]{};
  \draw (5.5,3.5) node [minimum size=.8cm,fill=lightgray,draw]{};
  \draw (5.5,5.5) node [minimum size=.8cm,fill=lightgray,draw]{};
  \draw (5.5,7.5) node [minimum size=.8cm,fill=lightgray,draw]{$\str Z_{3,4}$};

  \draw (6.5,7.5) node [minimum size=.8cm,fill=gray,draw]{};

  \draw (6,8) node [circle,fill=white,draw]{};

  \draw (6,7) node [circle,fill=black,draw]{};

 }

\end{tikzpicture}
\caption{$\str G_{7,8}$.}\label{fig:cyclecounting}
\end{figure}

\begin{proof}
Since $h$ is injective and $|C|= |F_1|$, it suffices to prove $C\subseteq
h(F_1)$. We consider all the 4-cycles
\[
\str Z_{i,j}: (2i-1, 2j-1) \to (2i-1, 2j) \to (2i, 2j) \to (2i, 2j-1) \to (2i-1, 2j-1),
\]
with $i\in [k_1]$ and $j \in [k_2+1]$. For the grid $\str G_{7,8}$ these are
the light gray cycles in Figure~\ref{fig:cyclecounting}. It is clear that
there are $k_1\cdot (k_2+1)$ of them, and none of them share common
vertices. Note $|C|= k_1\cdot (k_2+1)$. Thus Lemma~\ref{lem:C}~(C1) implies
that the $h$-image of each such cycle must contain exactly one vertex which
is mapped to $C$, and every vertex in $C$ is the image of one vertex in one
of those cycles.

Consider the $4$-cycle
\begin{align*}
(2k_1, 2k_2+1) \to (2k_1, 2k_2+2) \to & (2k_1+1, 2k_2+2) \\
 & \to (2k_1+1, 2k_2+1) \to (2k_1, 2k_2+1),
\end{align*}
which is the gray cycle in Figure~\ref{fig:cyclecounting} for our example
$\str G_{7,8}$. By~Lemma~\ref{lem:C}~(C1) it must contain a vertex $u$
mapped to $C$. As we have already argued that $u$ is on one of the cycles
$\str Z_{i,j}$, in particular $Z_{k_1,k_2+1}$. In addition, $Z_{k_1, k_2+1}$
can have only one such vertex $u$. We conclude that either $u= (2k_1,
2k_2+2)$ or $u=(2k_1, 2k_2+1)$.

Assume $u=(2k_1, 2k_2+2)$, i.e., $h(2k_1, 2k_2+2)\in C$, as shown by the
white node in Figure~\ref{fig:cyclecounting}. Observe that the 4-cycle (the
cycle left to $\str Z_{3,4}$)
\begin{align*}
(2k_1-2, 2k_2+1) \to (2k_1-2, & 2k_2+2) \to (2k_1-1, 2k_2+2) \\
 & \to (2k_1-1, 2k_2+1) \to (2k_1-2, 2k_2+1)
\end{align*}
also must contain at least one vertex mapped to $C$. Using
Lemma~\ref{lem:C}~(C2), we can conclude that it can only be $(2k_1-2,
2k_2+2)$. By repeating the argument inductively, we conclude that if
$h(i,j)\in C$, then both $i$ and $j$ have to be even.

In the second case, we have $u=(2k_1, 2k_2+1)$, i.e., $h(2k_1, 2k_2+1)\in C$
as shown by the black node in Figure~\ref{fig:cyclecounting}. Then $\bar
h(2k_1, 2) \in C$ by $h(2k_1, 2k_2+1)\in C$. Thus the same argument as the
above case shows that if $\bar h(i,j)\in C$, then both $i$ and $j$ have to
be even.
\end{proof}

Recall that our goal is to show that $\mathcal S_{s,t}=(F,D)$ is rigid for
$\str G_{s,t}$. In particular, $\big\{(u,u) \bigmid u\in F\big\} \subseteq
h(V(\str G_{s,t}))$. Since $h(V(\str G_{s,t}))= \bar h(V(\str G_{s,t}))$,
this is equivalent to
\[
\big\{(u,u) \bigmid u\in F\big\} \subseteq \bar h(V(\str G_{s,t})).
\]
So without loss of generality, in the following we assume that $h(F_1)=
C_1$. That is,
\begin{eqnarray}\label{eq:hC}
h(i,j) \in C & \iff & (i,j)\in F_1.
\end{eqnarray}

\begin{lemma}\label{lem:4corners}
\begin{align*}
 & \big\{h(2,2),  h(2, t), h(s-1, 2), h(s-1, t)\big\} \\
 = & \Big\{\big((2,2),(2,2)\big), \big((2, t),(2, t)\big), %\\ & \hspace{1cm}
 \big((s-1, 2),(s-1, 2)\big), \big((s-1, t),(s-1, t)\big)\Big\}.
\end{align*}
\end{lemma}

\begin{proof}
In the grid $\str G_{s, t}$ the only vertices in $F_1$ that are connected to
exactly two other vertices in $F_1$ by paths of length 2 are $(2,2)$, $(2,
t)$, $(s-1, 2)$, and $(s-1, t)$ (all the others are connected to either 3 or
4 vertices in $C$). By Lemma~\ref{lem:correctcenter}, $h$ has to map them to
those vertices in $C$ with the same property in $H$, which are precisely the
ones in the righthand side of the above equation.
\end{proof}

By similar arguments as the proof of Lemma~\ref{lem:gridframe}, in
particular the proof of Claim~2 based on distances, and by taking
automorphism if necessary, we can assume that for every $i\in [k_1]$ and
$j\in [k_2+1]$
\begin{equation}\label{eq:centers}
h(2i, 2j)= \big((2i,2j), (2i,2j)\big).
\end{equation}

\begin{lemma}\label{lem:gridcorners}
For every $u\in \big\{(1,1), (s,1), (1,t), (s,t)\big\}$ we have $h(u)=
(u,u)$.
\end{lemma}

\begin{proof}
We first show
\[
h(1,t)= \big((1,t), (1,t)\big).
\]
Since $h$ is an embedding and $h(2,t)= \big((2,t), (2,t)\big)$, it holds
that $h(2,t)$ is in
\[
\Big\{\big((1,t), (1,t)\big), \big((3,t), (3,t)\big), \big((2,t-1), (2,t-1)\big)\Big\},
\]
i.e., the set of vertices adjacent to $\big((2,t), (2,t)\big)$ in $P$.
Assume that $h(1,t)= \big((3,t), (3,t)\big)$. We consider the path
\[
(1,t) \to (2,t) \to (3,t) \to (4,t).
\]
Under the embedding $h$ we should get a path in $P$ as
\begin{align*}
h(1,t)= \big((3,t), (3,t)\big) \to h(2,t)& = \big((2,t), (2,t)\big) \\
& \to h(3,t) \to h(4,t)= \big((4,t), (4,t)\big),
\end{align*}
where the second and third equalities are by~\eqref{eq:centers}. But this
clearly forces $h(3,t)= \big((3,t), (3,t)\big)= h(1,t)$, which contradicts
the injectivity of $h$. The case for $h(2,t)= \big((2,t-1), (2,t-1)\big)$
can be similarly ruled out.

\medskip
Now we proceed to show that for all $j< t$ we have $h(1,j)= \big((1,j),
(1,j)\big)$. Let $j= t-1$. Since $(1,t)$ and $(1,t-1)$ are adjacent in $G$,
hence $h(1,t)= \big((1,t), (1,t)\big)$ and $h(1,t-1)$ are adjacent in $P$
too. Thus, $h(1,t-1)= \big((1,t-1), (1,t-1)\big)$. Furthermore, using
\begin{eqnarray*}
h(2,t)= \big((2,t),(2,t)\big)
 & \text{and} &
h(2,t-2)= \big((2,t-2),(2,t-2)\big)
\end{eqnarray*}
we deduce that $h(2, t-1)= \big((2, t-1), (2, t-1)\big)$. Combined with
$h(1,t-1)= \big((1,t-1),(1,t-1)\big)$, we conclude that
\[
h(1,t-2)= \big((1,t-2), (1,t-2)\big).
\]
Repeating the above argument, it can be reached that $h(1,1)= h(1,1)$.
Similarly we can obtain that $h(s,t)= h(s,t)$, and finally $h(s,1)= h(s,1)$.
\end{proof}

Now Proposition~\ref{prop:gridrigid} follows easily from
Lemma~\ref{lem:correctcenter} and Lemma~\ref{lem:gridcorners}.

%\begin{corollary}\label{cor:gridrich}
%Let $\cls K$ be a class of graphs such that for every $k\in \mathbb N$ there
%exists a grid $\str G_{s,t}\in \cls K$ with $\min \{s,t\}\ge k$. Then $\cls
%K$ is rich.
%\end{corollary}

\begin{proof}[Proof of Proposition~\ref{prop:gridwallrich}~(i)] Let $k\in \mathbb N$. Our goal
is to find a graph $G$ in $\cls K$ with a rigid skeleton $(F,D)$ such that
$\tw\big((G\setminus F)/D\big)\ge k$. Define
\[
k^*:= 2k+4.
\]
Thus there is a grid $\str G_{s,t}\in \cls K$ with $s,t\ge k^*$. By
Proposition~\ref{prop:gridrigid} the skeleton $\mathcal S_{s,t}= (F,D)$ in
Definition~\ref{defn:gridskeleton} of $\str G_{s,t}$ is rigid. Moreover,
\[
\tw\big((\str G_{s, t}\setminus F)/D\big)=
 \min\left\{\floor{\frac{s-1}{2}}, \floor{\frac{t-2}{2}}\right\}-1\ge k,
\]
by Lemma~\ref{lem:gridskeletontw}.
\end{proof}

%
%\medskip
%\begin{proof}[of Theorem~\ref{thm:main}] (1) is a direct consequence of
%Theorem~\ref{thm:richW1} and Corollary~\ref{cor:gridrich}. \qed
%\end{proof}

\subsection{The richness of walls}\label{subsec:walls}
We fix some $s> 2$ and $t> 3$. Let $\mathcal S^{\rm wall}_{s,t}:= (F, D)$
with $F= F_1\cup F_2 $ and $D= \emptyset$, where
\begin{equation}\label{eq:wallF}
\begin{array}{rl}
F_1 & = \big\{v_{i,j} \bigmid \text{$i\in [s+1]$ and $j\in [2]$}\big\} \\
 & \hspace{1.2cm}
 \cup\; \big\{u_{i,2} \bigmid i\in [s+1]\big\}
 \cup \big\{v_{i,t},u_{i,t} \bigmid i\in [s+1]\big\}\\[1mm]
 & \hspace{1.2cm}\cup\; \big\{v_{i,t+1} \bigmid \text{$i\in [s+1]$ and odd $t$}\big\}
  \cup \big\{u_{i,t+1} \bigmid \text{$i\in [s+1]$ and even $t$}\big\}, \\[2mm]
 F_2 &= \big\{v_{1,j},u_{1,j}, v_{2,j}\bigmid 3\le j< t\big\}
  \cup \big\{u_{s,j}, v_{s+1, j}, u_{s+1, j}\bigmid 3\le j <t\big\}.
\end{array}
\end{equation}
That is, we take as the set $F$ the bottom two rows, the top two rows, and
the leftmost three columns, and the rightmost three columns of the vertices
in $\str W_{s,t}$. Figure~\ref{fig:wallskeletion} shows the case for $\str
W_{6,7}$.
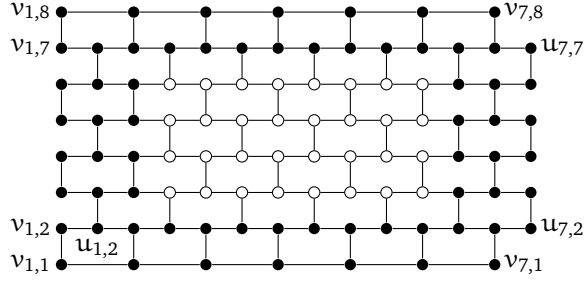
\begin{figure}%[h!]
\centering
\begin{tikzpicture}[
  scale=0.60,
  vertexD/.style={circle,inner sep=0pt,minimum size=1.5mm,fill=black},
  vertex/.style={circle,inner sep=0pt,minimum size=1.5mm,fill=white,draw}
  ]

\begin{scope}%[xshift=10cm,yshift=1cm]
  \foreach \x in {0,1,...,6}
      \foreach \y in {0,1,2,3}
      {
        \ifthenelse{\NOT \y=3}
        {
         \ifthenelse{\y=0 \OR \x=0 \OR \x=1 \OR \x=6}
          {\node[vertexD] (v\x\y) at ($(1.6*\x,1.6*\y)$) {};
           \node[vertexD] (vv\x\y) at ($(1.6*\x,1.6*\y)+(0,0.8)$) {};}
          {\node[vertex] (v\x\y) at ($(1.6*\x,1.6*\y)$) {};
           \node[vertex] (vv\x\y) at ($(1.6*\x,1.6*\y)+(0,0.8)$) {};};

         \ifthenelse{\y=0 \OR \x=0 \OR \x=5 \OR \x=6}
          {\node[vertexD] (u\x\y) at ($(1.6*\x,1.6*\y)+(0.8,0.8)$) {};}
          {\node[vertex] (u\x\y) at ($(1.6*\x,1.6*\y)+(0.8,0.8)$) {};}

         \ifthenelse{\y=2 \OR \x=0 \OR \x=5 \OR \x=6}
          {\node[vertexD] (uu\x\y) at ($(1.6*\x,1.6*\y)+(0.8,1.6)$) {};}
          {\node[vertex] (uu\x\y) at ($(1.6*\x,1.6*\y)+(0.8,1.6)$) {};}

         \draw (v\x\y) to (vv\x\y);
         \draw (u\x\y) to (uu\x\y);}
        {
         \node[vertexD] (v\x\y) at ($(1.6*\x,1.6*3)$) {};
         \node[vertexD] (vv\x\y) at ($(1.6*\x,1.6*\y)+(0,0.8)$) {};
         \draw (v\x\y) to (vv\x\y);
        }

        \ifthenelse{\y=0 \AND \NOT \x=0}
        {
        \pgfmathtruncatemacro{\xminusone}{\x - 1};
        \draw (v\x\y) to (v\xminusone\y);
        }{};

        \ifthenelse{\y<3}
        {
        \draw (vv\x\y) to (u\x\y);
        \ifthenelse{\x>0}
         {
         \pgfmathtruncatemacro{\xminusone}{\x - 1};
         \draw (vv\x\y) to (u\xminusone\y);
         }{};
        }{};

        \ifthenelse{\y>0}
        {
        \pgfmathtruncatemacro{\yminusone}{\y - 1};
        \draw (v\x\y) to (uu\x\yminusone);
        \ifthenelse{\x>0}
         {
         \pgfmathtruncatemacro{\xminusone}{\x - 1};
         \draw (uu\xminusone\yminusone) to (v\x\y);
         }{};
        }{};

        \ifthenelse{\y=3 \AND \NOT \x=0}
        {
        \pgfmathtruncatemacro{\xminusone}{\x - 1};
        \draw (vv\x\y) to (vv\xminusone\y);
        }{};

      }

      \path (v00) node[left] {\small$v_{1,1}$};
      \path (v60) node[right] {\small$v_{7,1}$};
      \path (vv00) node[left] {\small$v_{1,2}$};
      \path (u00) node[below] {\small$u_{1,2}$};
      \path (u60) node[right] {\small$u_{7,2}$};
      \path (v03) node[left] {\small$v_{1,7}$};
      \path (vv03) node[left] {\small$v_{1,8}$};
      \path (vv63) node[right] {\small$v_{7,8}$};
      \path (uu62) node[right] {\small$u_{7,7}$};

%     \path (uu02) node[left] {\small$u_{1,7}$};
%     \path (uu61) node[right] {\small$u_{7,5}$};

%     \path (5.2,-1) node{(b)}; % $(7\times 4)$-grid}};
\end{scope}

% the wall

%  \draw (11,1)--(17,1);

%  \draw (11,1)--(11,1.5);
%  \draw (12,1)--(12,1.5);
%  \draw (13,1)--(13,1.5);
%  \draw (14,1)--(14,1.5);
%  \draw (15,1)--(15,1.5);
%  \draw (16,1)--(16,1.5);
%  \draw (17,1)--(17,1.5);
% %%

%  \draw (11,1.5)--(17.5,1.5);

%  \draw (11.5,1.5)--(11.5,2);
%  \draw (12.5,1.5)--(12.5,2);
%  \draw (13.5,1.5)--(13.5,2);
%  \draw (14.5,1.5)--(14.5,2);
%  \draw (15.5,1.5)--(15.5,2);
%  \draw (16.5,1.5)--(16.5,2);
%  \draw (17.5,1.5)--(17.5,2);
% %%

%  \draw (11,2)--(17.5,2);

%  \draw (11,2)--(11,2.5);
%  \draw (12,2)--(12,2.5);
%  \draw (13,2)--(13,2.5);
%  \draw (14,2)--(14,2.5);
%  \draw (15,2)--(15,2.5);
%  \draw (16,2)--(16,2.5);
%  \draw (17,2)--(17,2.5);
% %%

%  \draw (11,2.5)--(17.5,2.5);

%  \draw (11.5,2.5)--(11.5,3);
%  \draw (12.5,2.5)--(12.5,3);
%  \draw (13.5,2.5)--(13.5,3);
%  \draw (14.5,2.5)--(14.5,3);
%  \draw (15.5,2.5)--(15.5,3);
%  \draw (16.5,2.5)--(16.5,3);
%  \draw (17.5,2.5)--(17.5,3);
% %%

%  \draw (11,3)--(17.5,3);

%  \draw (11,3)--(11,3.5);
%  \draw (12,3)--(12,3.5);
%  \draw (13,3)--(13,3.5);
%  \draw (14,3)--(14,3.5);
%  \draw (15,3)--(15,3.5);
%  \draw (16,3)--(16,3.5);
%  \draw (17,3)--(17,3.5);
% %%

%  \draw (11,3.5)--(17.5,3.5);

%  \draw (11.5,3.5)--(11.5,4);
%  \draw (12.5,3.5)--(12.5,4);
%  \draw (13.5,3.5)--(13.5,4);
%  \draw (14.5,3.5)--(14.5,4);
%  \draw (15.5,3.5)--(15.5,4);
%  \draw (16.5,3.5)--(16.5,4);
%  \draw (17.5,3.5)--(17.5,4);
% %%

%  \draw (11.5,4)--(17.5,4);

%  \draw (14,0) node{{\normalsize (b)}}; % $(6\times 6)$-wall.}};

\end{tikzpicture}
\caption{A skeleton $(F,D)$ for $\str W_{6,7}$, where $F$ is the set of black vertices and $D= \emptyset$.}\label{fig:wallskeletion}
\end{figure}

\medskip
The following observation is straightforward.
\begin{lemma}\label{lem:5cycle}
Every $5$-cycle in $\str W_{s,t}$ is contained in $F_1\subseteq F$.
\end{lemma}

Clearly those $5$-cycles are shortest odd cycles in $\str W_{s,t}$. The next
lemma explains their importance. Recall that a \emph{closed walk} in a graph
$G$ is a sequence
\[
v_1 \to v_2 \to \cdots v_k \to v_{k+1}=v_1
\]
of vertices such that $v_iv_{i+1}\in E(G)$ for every $i\in [k]$. Its length
is $k$. Thus, a $k$-cycle is a closed walk of length $k$ with $k\ge 3$ and
$v_i\ne v_j$ for every $1\le i< j\le k$.

\begin{lemma}\label{lem:oddclosedwalk}
Let $G$ be a graph and $k$ the length of a shortest odd cycle in $G$. Then
every closed walk in $G$ of length $k$ has to be a $k$-cycle.
\end{lemma}

\begin{proof}
Let
\begin{equation}\label{eq:closedwalk}
Z: v_1 \to v_2\to \cdots \to v_k \to v_1
\end{equation}
be a closed walk in $G$. We need to show that $v_i\ne v_j$ for every $1\le
i< j\le k$. Assume otherwise, choose such $i, j$ with $v_i=v_j$ and all
$v_k$'s in between pairwise distinct. Then either $j=i+2$ or
\[
Z_1: v_i \to v_{i+1} \to v_{i+2}\to \cdots \to v_j
\]
is a cycle. Note that $Z_1$ must be an even cycle, as no odd cycle in $G$
has length smaller than $k$. Thus in both cases, we can shorten the closed
walk~\eqref{eq:closedwalk} to
\[
v_1\to \cdots \to v_i \to v_{j+1} \to \cdots \to v_1
\]
which is still a closed walk of odd length. By repeating this procedure,
eventually we obtain a cycle of odd length in $G$, which contradicts the
minimality of $k$.
\end{proof}

\begin{corollary}\label{cor:oddcycle}
Let $G$ be a graph and $k$ the length of a shortest odd cycle in $G$. Then
every endomorphism of $G$ maps every $k$-cycle to a $k$-cycle.
\end{corollary}

\begin{proof} This is immediate by observing that every homomorphism maps a
cycle to a walk and by Lemma~\ref{lem:oddclosedwalk}.
\end{proof}

\begin{lemma}\label{lem:wallframe}
The set $F$ defined in~\eqref{eq:wallF} is a frame for the wall $\str
W_{s,t}$.
\end{lemma}

%\begin{proof}[Proof Sketch] Let $h$ be an endomorphism of $\str W_{s,t}$ with
%$F\subseteq h(V(\str W_{s,t}))$. Corollary~\ref{cor:oddcycle} implies that
%$h$ maps every 5-cycle in $\str W_{s,t}$ to a 5-cycle. By an easy induction,
%together with Lemma~\ref{lem:5cycle}, we can conclude that $h(F_1)= F_1$.
%Then a routine argument based on distance gives us the surjectivity of $h$.
%\end{proof}

\begin{proof}
We discuss the cases as exemplified in Figure~\ref{fig:wallskeletion} where
$s\le t$ and $t\ge 5$ is odd. The others can be argued in a similar fashion.

Let $h$ be an endomorphism of $\str W_{s,t}$ with $F\subseteq h(V(\str
W_{s,t}))$. We consider four pair of vertices
\begin{eqnarray*}
(v_{1,2}, v_{s+1,t}), (v_{1,2}, u_{s+1,t})
 & \text{and} &
(v_{1,t}, v_{s+1,2}), (v_{1,t}, u_{s+1,2}).
\end{eqnarray*}
It is not hard to see that they have the largest distance in $\str W_{s,t}$,
and the distance between any other pair of vertices is strictly smaller.
Figure~\ref{fig:wallcorners} illustrates the situation for $\str W_{6,7}$.
We proceed similarly as Claim~1 in the proof of Lemma~\ref{lem:gridframe} to
obtain
\begin{align*}
h\big(\{v_{1,2}, v_{s+1,t}, u_{s+1,t}\}\big)& = \{v_{1,2}, v_{s+1,t}, u_{s+1,t}\} \\
 & \text{and}\
h\big(\{v_{1,t}, v_{s+1,2}, u_{s+1,2}\}\big)= \{v_{1,t}, v_{s+1,2}, u_{s+1,2}\},
\end{align*}
or
\begin{align*}
h\big(\{v_{1,2}, v_{s+1,t}, u_{s+1,t}\}\big)& = \{v_{1,t}, v_{s+1,2}, u_{s+1,2}\} \\
 & \text{and}\
h\big(\{v_{1,t}, v_{s+1,2}, u_{s+1,2}\}\big)= \{v_{1,2}, v_{s+1,t}, u_{s+1,t}\}.
\end{align*}
Note that $v_{1,2}$, $v_{s+1,2}$, $v_{1, t}$ and $v_{s+1,t}$ are on
5-cycles, while $u_{s+1,2}$ and $u_{s+1,t}$ are not.
Corollary~\ref{cor:oddcycle} implies that $h$ maps every 5-cycle in $\str
W_{s,t}$ to a 5-cycle. Thus we can assume without loss of generality that
\begin{eqnarray*}
h(v_{1,2})= v_{1,2}
 & \text{and} &
h(v_{1,t})= v_{1,t},
\end{eqnarray*}
and hence $h(w)= w$ for any $w\in \big\{v_{s+1,2}, u_{s+1,2},v_{s+1,t},
u_{s+1,t}\big\}$. Once these six vertices are in the right place, by an easy
induction we can conclude that $h(w)= w$ for every $w\in F_1$.

Then we proceed to show $h(w)= w$ for every $w\in F_2$. As a first step,
consider the possibility of $h(u_{1,3})$. Note that $h(u_{1,2})= u_{1,2}$
(for $u_{1,2}\in F_1$) implies
\begin{equation}\label{eq:hu13}
h(u_{1,3})\in \big\{v_{1,2}, v_{2,2}, u_{1,3}\big\}.
\end{equation}
On the other hand,
\[
\dist^{\str W_{s,t}}(u_{1,3}, v_{1,t})
 \ge \dist^{\str W_{s,t}}\big(h(u_{1,3}), h(v_{1,t})\big)
 = \dist^{\str W_{s,t}}\big(h(u_{1,3}), v_{1,t}\big).
\]
With~\eqref{eq:hu13} we can conclude $h(u_{1,3})= u_{1,3}$.

Recall $v_{1,3}\in F_2\subseteq h\big(V(\str W_{s,t})\big)$. Choose an
arbitrary $w\in V(\str W_{s,t})$ with $h(w)= v_{1,3}$. We observe that for
every $z\in \big\{v_{1,2}, u_{s+1,2}, v_{1,t}, u_{s+1, t}\big\}\subseteq
F_1$
\[
\dist^{\str W_{s,t}}(v_{1,3},z) = \dist^{\str W_{s,t}}\big(h(w),h(z)\big)
 \le \dist^{\str W_{s,t}}(w,z).
\]
But this implies $w= v_{1,3}$.

For the remaining vertices in $F_2$, we can argue similarly. And finally by
observing the distance between any vertex in $V(\str W_{s,t})\setminus
(F_1\cup F_2)$ and $F_1\cup F_2$, we can establish the surjectivity of $h$.
\end{proof}

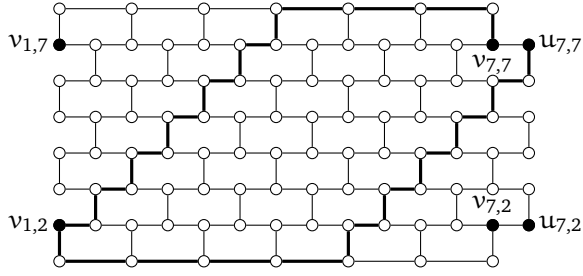
\begin{figure}%[h!]
\centering
\begin{tikzpicture}[
  scale=0.60,
  vertex/.style={circle,inner sep=0pt,minimum size=1.5mm,draw,fill=white},
  vertexD/.style={circle,inner sep=0pt,minimum size=1.5mm,draw,fill=black},
  ]

\begin{scope}[xshift=0cm,yshift=0cm]
  \foreach \x in {0,1,...,6}
      \foreach \y in {0,1,2,3}
      {
        \ifthenelse{\NOT \y=3}
        {
          \node[vertex] (v\x\y) at ($(1.6*\x,1.6*\y)$) {};
          \node[vertex] (vv\x\y) at ($(1.6*\x,1.6*\y)+(0,0.8)$) {};

          \node[vertex] (u\x\y) at ($(1.6*\x,1.6*\y)+(0.8,0.8)$) {};

          \node[vertex] (uu\x\y) at ($(1.6*\x,1.6*\y)+(0.8,1.6)$) {};
         \draw (v\x\y) to (vv\x\y);
         \draw (u\x\y) to (uu\x\y);}
        {
         \node[vertex] (v\x\y) at ($(1.6*\x,1.6*3)$) {};
         \node[vertex] (vv\x\y) at ($(1.6*\x,1.6*\y)+(0,0.8)$) {};
         \draw (v\x\y) to (vv\x\y);
        }

        \ifthenelse{\y=0 \AND \NOT \x=0}
        {
        \pgfmathtruncatemacro{\xminusone}{\x - 1};
        \draw (v\x\y) to (v\xminusone\y);
        }{};

        \ifthenelse{\y<3}
        {
        \draw (vv\x\y) to (u\x\y);
        \ifthenelse{\x>0}
         {
         \pgfmathtruncatemacro{\xminusone}{\x - 1};
         \draw (vv\x\y) to (u\xminusone\y);
         }{};
        }{};

        \ifthenelse{\y>0}
        {
        \pgfmathtruncatemacro{\yminusone}{\y - 1};
        \draw (v\x\y) to (uu\x\yminusone);
        \ifthenelse{\x>0}
         {
         \pgfmathtruncatemacro{\xminusone}{\x - 1};
         \draw (uu\xminusone\yminusone) to (v\x\y);
         }{};
        }{};

        \ifthenelse{\y=3 \AND \NOT \x=0}
        {
        \pgfmathtruncatemacro{\xminusone}{\x - 1};
        \draw (vv\x\y) to (vv\xminusone\y);
        }{};

      }

      %\node[vertex] (vv\x\y) at ($(1.6*\x,1.6*\y)+(0,0.8)$) {};
      \draw (0,0.8) node[vertexD]{};
      \path (vv00) node[left] {\small$v_{1,2}$};

      %\node[vertex] (vv\x\y) at ($(1.6*\x,1.6*\y)+(0,0.8)$) {};
      \draw ($(1.6*6,1.6*0)+(0,0.8)$) node[vertexD]{};
      \path (vv60) node[above] {\small$v_{7,2}$};

      \draw ($(1.6*6,1.6*0)+(0.8,0.8)$) node[vertexD] {};
      \path (u60) node[right] {\small$u_{7,2}$};

      \draw ($(1.6*0,1.6*3)$) node[vertexD]{};
      \path (v03) node[left] {\small$v_{1,7}$};

      \draw ($(1.6*6,1.6*3)$) node[vertexD]{};
      \path (v63) node[below] {\small$v_{7,7}$};

      \draw ($(1.6*6,1.6*2)+(0.8,1.6)$) node[vertexD] {};
      \path (uu62) node[right] {\small$u_{7,7}$};

      \draw[very thick] (vv00) -- (v00) -- (v10) -- (v20)
       -- (v30) -- (v40) -- (vv40) -- (u40) -- (uu40)
       -- (v51) -- (vv51) -- (u51) -- (uu51)
       -- (v62) -- (vv62) -- (u62) -- (uu62);

      \draw[very thick] (vv00) -- (u00) -- (uu00)
        -- (v11) -- (vv11) -- (u11) -- (uu11)
        -- (v22) -- (vv22) -- (u22) -- (uu22)
        -- (v33) -- (vv33) -- (vv43) -- (vv53) -- (vv63) -- (v63);

\end{scope}

\end{tikzpicture}
\caption{Two shortest paths between
$v_{1,2}$ and $v_{7,7}$ and between $v_{1,2}$ and $u_{7,7}$ are shown in black.}\label{fig:wallcorners}
\end{figure}

\begin{lemma}\label{lem:oddcycle}
Assume:
\begin{itemize}
\item $G$ is a graph with a skeleton $\mathcal S= (F,D)$.

\item $H$ is a graph with a coloring $\chi: V(H)\to V(G)\setminus (F\cup
    D)$.

\item $P= P(G, \mathcal S, H, \chi)$ is the product graph as defined in
    Definition~\ref{def:product}.

\item $k$ is the length of a shortest odd cycle in $G$.
\end{itemize}
Then for every cycle $Z$ in $P$ of length $k$, the set
\[
\big\{v \in V(G) \bigmid \text{$(v,a)$ occurs in $Z$ for some $a$}\big\}
\]
induces a cycle in $G$ of length $k$.
\end{lemma}

\begin{proof}
Assume $Z: (v_1,a_1) \to (v_2, a_2)\to \cdots \to (v_k, a_k) \to
(v_1, a_1)$. By Lemma~\ref{lem:pi1hom} we conclude that
\begin{equation}\label{eq:walk}
\pi_1(Z): v_1\to v_2 \to \cdots \to v_k \to v_1
\end{equation}
is a closed walk in $G$ of length $k$. Then the result follows immediately
from Lemma~\ref{lem:oddclosedwalk}.
\end{proof}

\begin{proposition}\label{prop:wallrigid}
$\mathcal S^{\rm wall}_{s,t}$ is a rigid skeleton of $\str W_{s,t}$.
\end{proposition}

\begin{proof}%[Proof Sketch]
Let $H$ be a graph and $\chi: V(H)\to V(\str W_{s,t})\setminus F$ where
$\mathcal S^{\rm wall}_{s,t}= (F,\emptyset)$. Moreover let $P= P(\str
W_{s,t}, \mathcal S^{\rm wall}_{s,t}, H, \chi)$ and $h$ be an embedding from
$\str W_{s,t}$ to $P$. We need to show that
\begin{equation}\label{eq:wallrigid}
\big\{(u,u) \bigmid u\in F\big\}\subseteq h(V(\str W_{s,t})).
\end{equation}
Again those 5-cycles $Z$ in $\str W_{s,t}$ plays a vital role. Since $h$ is
an embedding, the image $h(Z)$ has to be a 5-cycle in $P$. Hence
Lemma~\ref{lem:oddcycle} implies that $\pi_1(h(Z))$ remains a 5-cycle in
$\str W_{s,t}$. Then $h(F_1)= \big\{(u,u) \bigmid u\in F_1\big\}$ by an easy
counting, and~\eqref{eq:wallrigid} follows again by observing the distance
between $F_1$ and $F_2$.
\end{proof}

Now the remaining part of the proof of the richness of walls \big(i.e.,
Proposition~\ref{prop:gridwallrich}~(ii)\big) is easy, and thus left to the
reader.
%
%\begin{cor}\label{cor:wallrich}
%Let $\cls K$ be a class of graphs such that for every $k\in \mathbb N$ there
%exists a grid $\str W_{s,t}\in \cls K$ with $\min \{s,t\}\ge k$. Then $\cls
%K$ is rich.
%\end{cor}

\section{Conclusions}\label{sec:con}

We have shown that the parameterized embedding problem on the classes of all
grids and all walls is hard for \W 1. Our proof exploits some general
structures in those graphs, i.e., frames and skeletons, thus is more generic
than other known \W 1-hard cases. We expect that our machinery can be used
to solve some other cases. However, it could be seen that the class of
complete bipartite graphs is not rich. Hence the result of~\cite{lin15} is
not a special case of our Theorem~\ref{thm:richW1}. Resolving the
\textbf{Dichotomy Conjecture} for the embedding problem might require a
unified understanding of the cases of biclique and grids.

A remarkable phenomenon of the homomorphism problem is that the polynomial
time decidability of $\Hom(\cls K)$ coincides with the fixed-parameter
tractability of $\pHom(\cls K)$ for any class $\cls K$ of
graphs~\cite{gro07}, assuming $\FPT\ne \W 1$. For the embedding problem this
is certainly not true, as for the class $\cls K$ of all paths $\Emb(\cls K)$
is \NP-hard, yet $\pEmb(\cls K)\in \FPT$. Thus, in the \textbf{Dichotomy
Conjecture}, the tractable side is really in terms of fixed-parameter
tractability. But it is still interesting and important to give a precise
characterization of those $\cls K$ whose $\Emb(\cls K)$ are solvable in
polynomial time. At the moment, we don't even have a good conjecture.

%
% ---- Bibliography ----
%

\bibliographystyle{plain}
\bibliography{grid}

%%%%%%%%%%%%%%%%%%%%%%%%%%%%%%%%%%%%%

%\newpage

\end{document}